\documentclass[11pt]{article}
\usepackage[usenames,dvipsnames,svgnames,table]{xcolor}
\usepackage{xparse}
\usepackage{tikz}
\usetikzlibrary{positioning,patterns,matrix,arrows,calc}
\usetikzlibrary{decorations.markings}
\usetikzlibrary{patterns}
\usetikzlibrary{snakes}
\usetikzlibrary{shapes}
\usepackage[latin1]{inputenc}
\usetikzlibrary{trees}
\usetikzlibrary{decorations.pathmorphing}
\usepackage{amsfonts, simplewick}
\usepackage{amssymb}
\usepackage{amscd}
\usepackage{amsmath}
\usepackage{epsfig}
\usepackage{graphicx, subfigure}
\usepackage{latexsym}
\usepackage[font=small,labelfont=bf]{caption} 
\usepackage{verbatim}
\usepackage{tipa}
\usepackage{amsthm,amssymb}
\usepackage{color}
\newlength{\dinwidth}
\newlength{\dinmargin}
\newcommand{\be}{\begin{equation}}
\newcommand{\ee}{\end{equation}}
\newcommand{\ba}{\begin{eqnarray}}
\newcommand{\ea}{\end{eqnarray}}

\newtheorem{definition}{Definition}
\newtheorem{theorem}{Theorem}

\newtheorem{remark}{Remark}

\newtheorem{example}{Example}

\def\e{{e}}
\begin{document}
\title{Enumeration of $N$-rooted maps using quantum field theory}
\author{K. Gopala Krishna$^1$, Patrick Labelle$^2$, and Vasilisa Shramchenko$^3$}
\date{}
\maketitle
\footnotetext[1]{Department of mathematics and statistics, Concordia University, E-mail: {\tt gopala.krishna@concordia.ca}}
\footnotetext[2]{Champlain Regional College, Lennoxville campus, Sherbrooke, Quebec, Canada. E-mail: {\tt plabelle@crc-lennox.qc.ca}}
\footnotetext[3]{Department of mathematics, University of
Sherbrooke, 2500, boul. de l'Universit\'e,  J1K 2R1 Sherbrooke, Quebec, Canada. E-mail: {\tt  Vasilisa.Shramchenko@Usherbrooke.ca}}
 
 \vspace{-0.7cm}
 
 \begin{abstract}
A one-to-one correspondence is proved between the $N$-rooted ribbon graphs, or maps,  with $\e$ edges and the $(\e-N+1)$-loop Feynman diagrams of a certain quantum field theory. This result is used to obtain explicit expressions and relations for the generating functions of $N$-rooted maps and for the numbers of $N$-rooted maps with a given number of edges using  the path integral approach applied to the corresponding quantum field theory.  
 \end{abstract}

\tableofcontents

%AMS subject classification: 05C30, 81Q30, 81T18

   \tikzset{
    photonloop/.style={decorate, decoration={snake,amplitude=1.5pt,segment length=4 pt},draw=black},
        photon/.style={decorate, decoration={snake},draw=black},
    electron/.style={draw=black, postaction={decorate},
        decoration={markings,mark=at position .75 with {\arrow[draw=black,scale=2]{>}}}},
         electronf/.style={draw=black, postaction={decorate},
        decoration={markings,mark=at position .15 with {\arrow[draw=black,scale=2]{>}}}},
              electronff/.style={draw=black, postaction={decorate},
        decoration={markings,mark=at position .25 with {\arrow[draw=black,scale=2]{>}}}},
          electronfff/.style={draw=black, postaction={decorate},
        decoration={markings,mark=at position .85 with {\arrow[draw=black,scale=2]{>}}}},
          electronmiddle/.style={draw=black, postaction={decorate},
        decoration={markings,mark=at position .55 with {\arrow[draw=black,scale=2]{>}}}},
    positron/.style={draw=black, postaction={decorate},
        decoration={markings,mark=at position .55 with {\arrow[draw=black,scale=2]{<}}}},
    gluon/.style={decorate, draw=magenta,
        decoration={coil,amplitude=4pt, segment length=5pt}},
        fermion/.style={draw=black, postaction={decorate},decoration={markings,mark=at position .55 with {\arrow[draw=black,scale=2]{>}}}},
  vertex/.style={draw,shape=circle,fill=black,minimum size=3pt,inner sep=0pt}
}

\NewDocumentCommand\semiloop{O{black}mmmO{}O{above}}
{
\draw[#1] let \p1 = ($(#3)-(#2)$) in (#3) arc (#4:({#4+180}):({0.5*veclen(\x1,\y1)})node[midway, #6] {#5};)
}

\NewDocumentCommand\semiloopt{O{black}mmmO{}O{above}}
{
\draw[#1] let \p1 = ($(#3)-(#2)$) in (#3) arc (#4:({#4+245}):({0.5*veclen(\x1,\y1)})node[midway, #6] {#5};)
}

\NewDocumentCommand\semilooptt{O{black}mmmO{}O{above}}
{
\draw[#1] let \p1 = ($(#3)-(#2)$) in (#3) arc (#4:({#4+285}):({0.3*veclen(\x1,\y1)})node[midway, #6] {#5};)
}

\section{Introduction}

\noindent Enumeration of rooted maps started with the work by Tutte \cite{Tutte1, Tutte2} on counting planar maps, followed by \cite{WalshLehman1, WalshLehman2, WalshLehman3} and \cite{JacksonVisentin, Bender, Bender2} in arbitrary genus, and more recent results, see \cite{ArquesBeraud, Chapuy, Chapuy2} and references therein.  Beginning with the seminal work of t'Hooft \cite{thooft} on the applications of matrix integrals to Yang-Mills gauge theories and in particular quantum chromodynamics, maps have become a major tool in  quantum field theory and in string theory. See for example \cite{Eynard} for a  review on matrix models and the enumeration of maps.

\vskip 0.2cm

\noindent In this article, we consider rooted ribbon graphs, that is graphs embedded into a compact oriented surface in such a way that each face is a topological disc and one half-edge is distinguished, that can also be seen as rooted maps. The main aim of this article is to introduce $N$-rooted ribbon graphs extending the notion of rooted ribbon graphs, where $N$ distinct vertices of the graph are rooted, and to solve the enumeration problem for these graphs. This is a continuation of \cite{ArquesBeraud, Tutte1, Tutte2,  WalshLehman1, WalshLehman2, WalshLehman3} where the corresponding enumeration problem for $1$-rooted ribbon graphs was solved.
% \color{blue}
 Our main idea is to apply methods of quantum field theory to enumeration of graphs. This is possible due to the bijection that we establish between $N$-rooted maps and Feynman diagrams for $2N$-point function in a quantum field theory of two interacting scalar (spin 0) fields, as defined in section \ref{sect_Wick}. We will refer to this theory as scalar quantum electrodynamics (scalar QED) to follow the notation of \cite{prd},  even though this is an abuse of language because our theory does not contain a spin one gauge field.
% This bijection, proved in this paper, was first noticed in \cite{fdiag} for the case of one-rooted maps, where it was verified up to the $3$-loop level for two-point function. 
\color{black}

\vskip 0.2cm

\noindent The 
%\color{blue}
 coincidence of \color{black} the number of two-point Feynman diagrams with regard to the perturbative order in scalar QED  and the number of rooted maps as a function of the number of edges has already been observed and in \cite{fdiag} an intuitive association between the two objects was proposed. This was verified up to third order by starting from the Feynman diagrams and using the proposed association to generate the corresponding rooted maps. However, the formal proof of a bijection between the two classes of objects to all orders has not yet appeared. Verifying the bijection explicitly to higher order becomes impractical very quickly owing to the rapid increase in the number of Feynman diagrams and rooted maps at higher orders. For instance, the number of Feynman diagrams with $4$ loops or rooted maps with $4$ edges is $706$, while at $5$ loops or edges there are $8162$ diagrams or graphs. 

%\color{blue}

\vskip 0.2cm

\noindent In this paper, we prove the equality of the number of two-point Feynman diagrams in  scalar QED and the number of rooted maps in two ways. First we  notice that the differential equation for the number of rooted maps as function of the number of edges derived in \cite{ArquesBeraud} coincides with the differential equation from quantum field theory that govern the number of two-point Feynman diagrams. Our second proof establishes the direct bijective correspondence between Feynman diagrams in question and the rooted maps by using Wick's theorem from quantum field theory. \color{black} The use of Wick's theorem and the technique of ribbon graphs turns out to be the formalization of the correspondence put forward in \cite{fdiag}. 
%\color{blue}
 It is this second proof that admits a generalization to the general case of $2N$-point Feynman diagrams and leads naturally to the definition of $N$-rooted maps. We thus prove the correspondence observed in \cite{fdiag} and generalize it to the bijection between $2N$-point Feynman diagrams and $N$-rooted maps; this is one of the main results of the present article and the statement of Theorem \ref{thm_bijection}. 

\vskip 0.2cm

\noindent We would like to point out that the correspondence between rooted ribbon graphs and Feynman diagrams that we establish here is very different from the one obtained via the matrix model approach and widely exploited after the seminal work of 't Hooft \cite {thooft}. In \cite{A} and \cite{B} an approach similar to ours was used for the vacuum diagrams of the QED theory we consider here.

\vskip 0.2cm

\noindent The result of Theorem \ref{thm_bijection} allows us to use the methods of quantum field theory to enumerate $N$-rooted graphs. Thus we  find the number $m_N(\e)$ of $N$-rooted graphs with given number $\e$ of edges by first solving the corresponding enumeration problem for connected $2N$-point Feynman diagrams along the lines of \cite{prd}. This appears to be a very powerful approach as it reproduces the result of \cite{ArquesBeraud} for enumeration of one-rooted maps ($N=1$) with little effort as explained in  Section \ref{sect_onerooted_count}. In addition to bypassing the laborious derivation from \cite{ArquesBeraud} of the formula for the number $m_1(\e)$ by recursively reconstructing rooted graphs from simpler rooted graphs, our approach yields a closed form expression for generating functions of numbers of more general $N$-rooted graphs. 

\vskip 0.2cm

\noindent More precisely, introducing the generating function for the numbers $m_N(\e)$ by
\begin{equation*}
M_N(\lambda) = \sum_{\e=0}^\infty m_N(\e)\, \lambda^{2\e}\,,
\end{equation*}
we find the second main result of the paper, formulated in Theorem \ref{thm_thm2}, namely the following closed form algebraic expression for these functions with $N\geq 1$: 
\begin{equation*}  
  M_N (\lambda) =  \sum_{\substack{\alpha_1 + 2 \alpha_2 + \ldots +N \alpha_N = N \\ \alpha_1 \ldots \alpha_N \geq 0 } }~ ~
  \frac{N!}{\alpha_1! \alpha_2 ! \ldots \alpha_N!} ~ 
    \frac{ (-1)^{\alpha_1 + \ldots +\alpha_N -1} ~ \left( \alpha_1 + \ldots  +\alpha_{N}-1 \right)!}{{\cal{Z}}_0^{\alpha_1 + \ldots + \alpha_N}} 
   ~  \prod_{1 \leq j \leq N}  \left( \frac{{\cal{Z}}_{j}}{  (j!)^2} \right)^{\alpha_j}\, , 
\end{equation*}
where
\begin{equation*}
{\cal{Z}}_j =  \sum_{k=0}^\infty ~ \frac{(2k+j)! \, (2k-1)!!}{(2k)!} ~ \lambda^{2k} \,, \qquad j\geq 0. 
\end{equation*}

\vskip 0.2cm

\color{black}

\noindent Finally, in Theorem $3$ it is shown that the generating functions $M_N(\lambda)$ for $N$-rooted maps are degree $N$ polynomial expressions with $\lambda$-dependent coefficients in the $M_1(\lambda)$. 

\color{black}
\noindent The paper is organized as follows. In Section \ref{sect_definitions} we collect known definitions relevant to rooted ribbon graphs as well as introduce the definition of $N$-rooted ribbon graphs. We also define a generating function $M_N$ of the numbers of $N$-rooted ribbon graphs with a given number of edges and review relevant known results on the generating function in the case $N=1$. In Section \ref{sect_Wick} we describe the class of Feynman diagrams studied in the paper and explain the statement of Wick's theorem allowing to generate all possible diagrams of the quantum field theory in question. In Section \ref{sect_bijection} we prove the bijection between the Feynman diagrams of our quantum field theory with $N$ external electron lines and $l$ loops on one side and $N$-rooted ribbon graphs with $l+N-1$ edges on the other.  In Section \ref{sect_QFT} we explain basic principles of the  path integration approach for the relevant quantum field theory. In Section \ref{sect_count} we re-introduce the generating functions for Feynman diagrams using the path integration technique. In Section \ref{sect_onerooted_count} we apply the theory developed in two preceding sections to rederive the known results of \cite{ArquesBeraud} on the generating function for one-rooted maps. In Section \ref{sect_Nrooted_count} we derive a closed form expression for the generating function of $N$-rooted ribbon graphs, or $N$-rooted maps, using the technique of path integration of the described quantum field theory. Furthermore, in the case $N=2$ and $N=3$ we do the calculation leading to a closed form formula for the number of $N$-rooted maps as a function of the number of edges. This calculation presents an algorithm that can be extended to an arbitrary given $N$ in a straightforward way. Finally, in Section \ref{sect_equation} we prove that the generating function $M_N(\lambda)$ of the numbers of $N$-rooted maps can be expressed as a degree $N$ polynomial in $M_1(\lambda)$. In the Appendix, we show a quantum field theory derivation of a differential equation which plays a central role in the proof of Theorem 3.

\noindent {\bf Acknowledgements.} 
K.G. wishes to thank Dmitry Korotkin, Marco Bertola as well as the staff and members at Concordia University for the support extended during his stay.  P.L. gratefully acknowledges the support from  the Fonds  de recherche du Qu\'ebec - Nature et technologies (FRQNT)  via a grant from the Programme de recherches pour les chercheurs de coll\`ege and to the STAR research cluster of  Bishop's University.  V.S. is grateful for the support from the Natural Sciences and Engineering Research Council of Canada through a Discovery grant as well as from the University of Sherbrooke. P.L. and V.S.  thank the  Max Planck Institute for Mathematics in Bonn, where this work was initiated, for hospitality and a perfect working environment.

\section{$N$-rooted graphs}
\label{sect_definitions}

A map is a cellular graph, that is a graph embedded into a connected compact orientable surface in such a way that each face is homeomorphic to an open  disc. The orientation of the underlying surface leads to a cyclic (counterclockwise) ordering on the half-edges incident to each vertex of a map. The notion of a map is equivalent to that of a ribbon graph; we use these two terms interchangeably.

\label{definitions}
\begin{definition}
A \emph{ribbon graph}, or a map, is the data $\Gamma = (H, \alpha,\sigma)$ consisting of  a set of half-edges $H = \{1,\dots, 2e\}$ with $e$ a positive integer and two permutations $\alpha, \sigma \in S_{2e}$ on the set of half-edges such that
	\begin{itemize}
	\item $\alpha$ is a fixed point free involution,
	\item the subgroup of $S_{2e}$ generated by $\alpha$ and $\sigma$ acts transitively on $H$. 
	\end{itemize}
\end{definition}

\noindent The involution $\alpha$  is a set of transpositions each of which pairs two half-edges that form an edge. Cycles of the permutation $\sigma$ correspond to vertices of the ribbon graph $\Gamma$; each cycle gives the ordering of half-edges at the corresponding vertex. Cycles of the permutation $\sigma^{-1}\circ\alpha$ correspond to faces of $\Gamma$. The condition of transitivity of the group $\langle \sigma, \alpha \rangle$ on the set of half-edges ensures the connectedness of the graph $\Gamma$. 

\vskip 0.2cm

\noindent A ribbon graph defines a connected compact orientable surface. This surface is reconstructed by gluing discs to the faces of the ribbon graph. The genus of the surface is called the genus of the ribbon graph.  Let us denote the set of vertices of a map (a ribbon graph) $\Gamma = (H, \alpha,\sigma)$ by $V$ and the set of faces by $F$. Recall that the set $V$ is in a bijection with the set of cycles of the permutation $\sigma$ and the set $F$ is in a bijection with the cycles of $\sigma^{-1}\circ\alpha$. Associated to each map is its {\it Euler characteristic} defined by 
\begin{equation*}
\chi(\Gamma) = \mid V\mid - \mid E\mid + \mid F\mid.
\end{equation*}
\noindent The Euler characteristic of a map is an invariant of the map, it depends only on the genus $g$ of the map and is given by $\chi(\Gamma) = 2-2g$. \\

\begin{definition}
\label{def_ribbon_iso}
An isomorphism between two ribbon graphs $\Gamma = (H, \alpha,\sigma)$ and $\Gamma' = (H, \alpha',\sigma')$ is a permutation $\psi\in S_{2e}$, that is $\psi:H\to H$, such that $\alpha'\circ \psi = \psi\circ \alpha$ and $\sigma'\circ \psi = \psi\circ \sigma$.
\end{definition}

\noindent Two isomorphic ribbon graphs are identified.  For a given graph $\Gamma=(H, \alpha,\sigma)$, the automorphisms are permutations on the set of half-edges, $\psi\in S_{2e}\,$ which commute with $\sigma$ and $\alpha$.
% isomorphisms of $\Gamma$ with itself are its automorphisms and the group of such automorphisms is denoted by ${\rm Aut}(\Gamma)$. 

\vskip 0.2cm

\noindent In terms of embeddings into a surface, two maps are equivalent if they can be transformed one into another by a homeomorphism of the underlying surface. For example, two maps $1$ and $2$ in Figure \ref{fig:torus} are equivalent, both corresponding to the ribbon graph $3$ from the same figure.

  \begin{figure}
\centering
\includegraphics[scale=0.4]{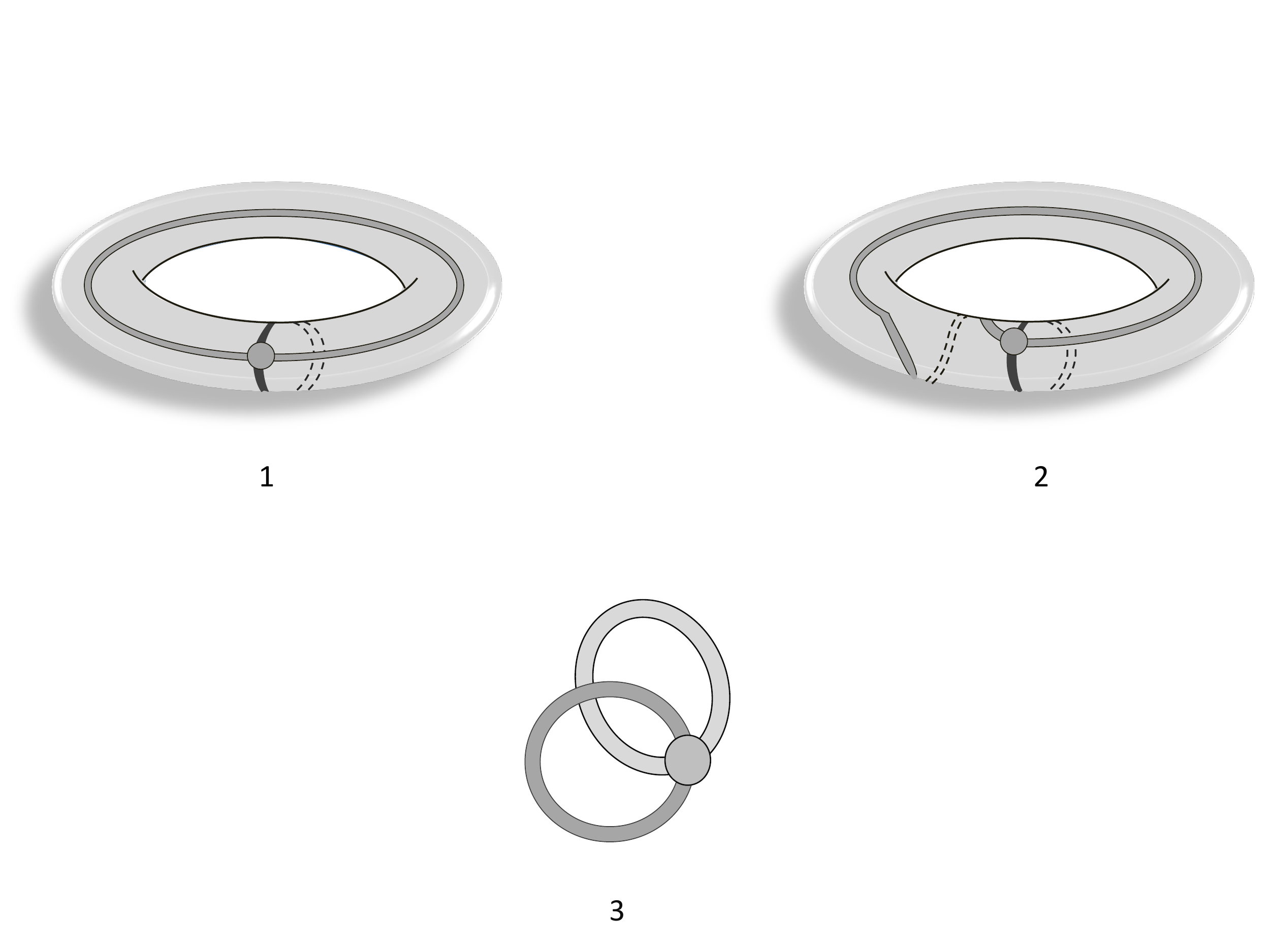}
\caption{A genus one map.}
\label{fig:torus}
\end{figure}

\noindent Automorphisms of ribbon graphs make their enumeration difficult and a way to deal with this is to introduce the idea of a rooted ribbon graph. 
\begin{definition}
\label{def_rooted}
A {\rm rooted  graph} is a ribbon graph with a distinguished half-edge, 
the {\rm root} of the graph. The  vertex to which the root is incident is called the 
 \emph{root vertex}.
\end{definition}
\begin{remark}
If there are no half-edges, the graph consists of one point. This is also considered as a rooted graph. 
\end{remark}
\noindent 
An isomorphism between two rooted ribbon graphs is an isomorphism of the ribbon graphs that maps the root to the root, a {\it rooted isomorphism} between the two graphs. 
For a given rooted graph, the group of its rooted automorphisms is trivial, see \cite{WalshLehman1}. 
\vskip 0.2cm
 
\noindent We propose the following generalization of the concept of a rooted graph. 
\begin{definition}
\label{def_Nrooted}
An $N$-rooted graph is the data of a ribbon graph, $\Gamma = (H, \alpha, \sigma)$, with the choice of $N$ distinct ordered elements    $\hat h_1,\dots, \hat h_N$ of $H$, called \emph{root half-edges}, or \emph{roots}, belonging to $N$ distinct cycles of $\sigma$, that is incident to $N$ distinct ordered vertices, called \emph{root vertices}.
\end{definition}

\noindent In other words, an $N$-rooted graph is obtained from a ribbon graph by choosing $N$ distinct vertices, labelling them with numbers from $1$ to $N$, and at each of the chosen vertices placing an arrow on one of the half-edges incident to it. In Figure \ref{fig:rooted}   examples of  two-rooted graphs of genus zero are shown. \\

  \begin{figure}
\centering
\includegraphics[scale=0.4]{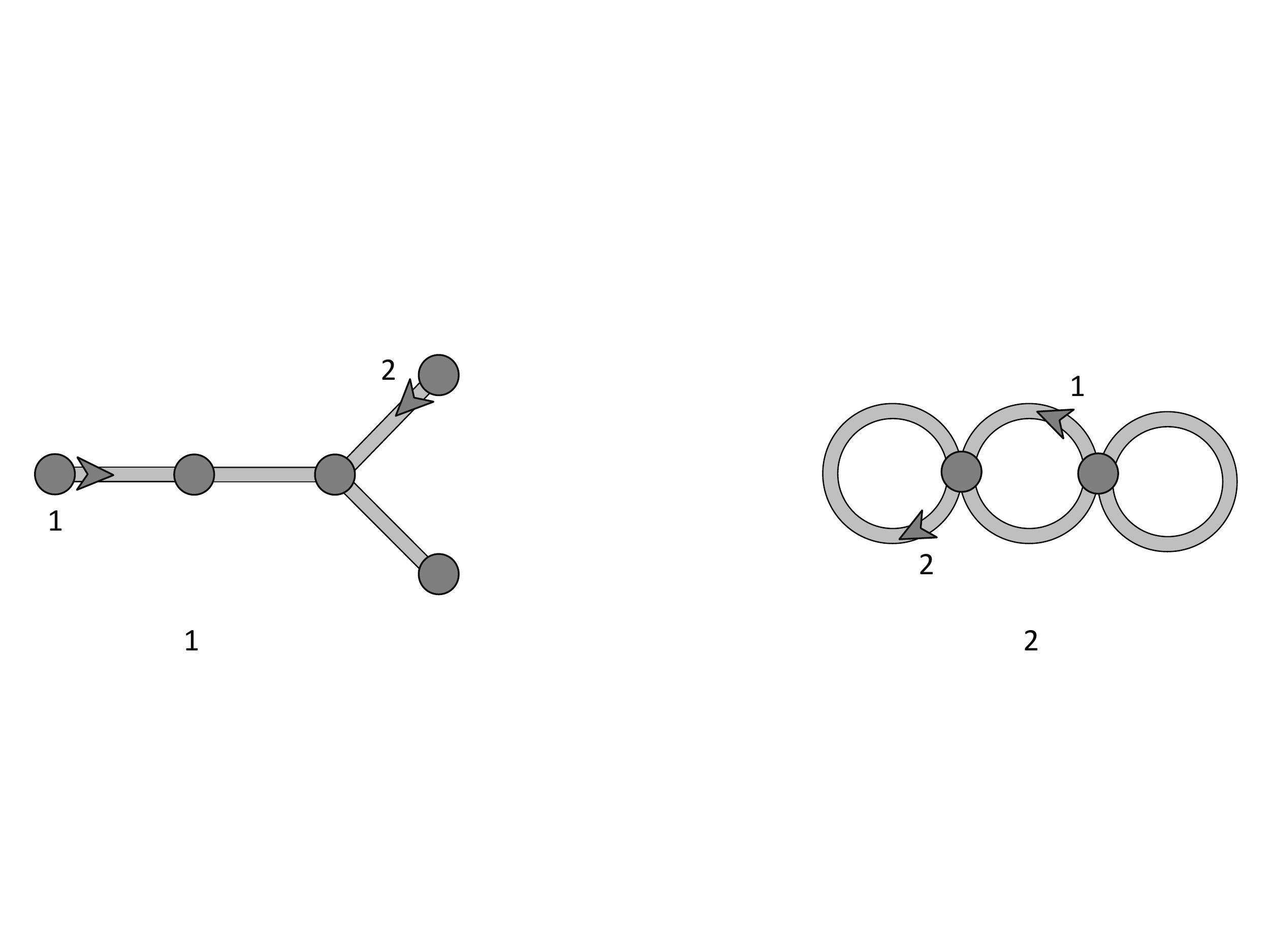}
\caption{Two-rooted graphs.}
\label{fig:rooted}
\end{figure}

\noindent 
An isomorphism between two $N$-rooted ribbon graphs is an isomorphism of the ribbon graphs that preserves the labelling of the $N$ root vertices and maps roots to roots.
We call such an isomorphism an {\it $N$-rooted isomorphism} between the two graphs. Similarly, an $N$-rooted automorphism of an $N$-rooted graph is an automorphism of the underlying ribbon graph which preserves the set of $N$ root vertices pointwise and maps roots to roots. Clearly, the only $N$-rooted automorphism of an $N$-rooted graph is the identity.

\subsection{Enumerating rooted ribbon graphs}
\label{sect_M1}
We are interested in computing the number of $N$-rooted ribbon graphs and
 for that it is useful  to arrange the graphs by the number of edges and to define the generating function $M_N(\lambda)$   such that 
\begin{equation}
M_N(\lambda) = \sum_{\e=0}^\infty m_N(\e)\, \lambda^{2\e}, \label{mnz}
\end{equation}
where $m_N(\e)$ is the number of $N$-rooted graphs with $\e$ edges regardless of genus. 
 The generating function $ M_1(\lambda)$ of the single rooted ribbon graphs  has been well studied in the literature, see \cite{ArquesBeraud}. From that reference, we have

\begin{equation*}
M_1(\lambda) = 1 + 2 \, \lambda^2 + 10 \,  \lambda^4 + 74  \, \lambda^6 + 706  \, \lambda^8+ 8 \, 162  \, \lambda^{10} + 110 \, 410 \, \lambda^{12} + \ldots \;.
\end{equation*}

\noindent Arqu\`es and B\'eraud \cite{ArquesBeraud}  also showed that the generating function for one-rooted graphs satisfies the following differential equation\footnote{The different form of this equation in \cite{ArquesBeraud}, namely $ M_1(z) = 1 + z  M_1(z)^2 + z  M_1(z) + 2 z^2 \ \dfrac{\partial  M_1(z)}{\partial z},$ is due to their use of another variable, $z=\lambda^2$,  in definition \eqref{mnz} of the generating function.  }:

\begin{equation}\label{M1DifferentialEquation}
 M_1(\lambda) = 1 + \lambda^2  M_1(\lambda)^2 + \lambda^2  M_1(\lambda) + \lambda^3 \ \dfrac{\partial  M_1(\lambda)}{\partial \lambda}\;.
\end{equation}

\noindent This equation is similar to the one obtained by Tutte for planar maps \cite{Tutte1}. It is a recursive relation where one-rooted ribbon graphs with a given number of edges are constructed recursively from one-rooted ribbon graphs with fewer edges. 
\noindent The above recursive relation leads to an expression for the number of rooted maps with $\e$ edges \cite{ArquesBeraud}: 
\begin{equation}\label{M1e}
m_1(\e) = \frac{1}{2^{\e+1}} \sum_{i=0}^{\e} (-1)^i \sum_{\substack{k_1+\cdots + k_{i+1}= \e+1 \\ k_1, \ldots, k_{i+1} > 0}} \prod_{j=1}^{i+1}  \dfrac{(2k_j)!}{k_j!}. 
\end{equation}

\section{ Feynman diagrams and Wick's theorem}
\label{sect_Wick}

\noindent Feynman diagrams are graphical tools that represent physical processes occurring in quantum field theories. The diagrams depend on the quantum field theory in question and the field content of the theory. To each field in the quantum field theory is associated a distinct type of edge, and interactions between the various kinds of fields are captured by vertices.
\vskip 0.2cm
  
 \noindent The particular quantum field theory we need in the present work involves a neutral scalar ({\it{i.e.}}  spin zero) field $A$ and a complex charged scalar field $\phi$ whose complex conjugate  will be denoted by $\phi^*$. This can be thought  of as a scalar version of quantum electrodynamics (QED) and by  abuse of language we will refer to the $A$ field as the photon field and to the $\phi$ field as the electron field, following \cite{prd}.
  
  \vskip 0.2cm
  
\noindent In our theory the propagation of the photon field will be denoted by a wavy line, the propagation of an electron by a solid, oriented line,  and there will be only one type of vertex where one photon line ends on  an electron line. An electron line can either be a loop oriented counterclockwise or a vertical line extending to infinity on both ends  and oriented upwards. The lines that extend to infinity will be referred to as external lines and the others will be referred to as internal lines.  Any deformations of the photon and electron lines preserving the orientations without cutting them or having a vertex pass through another vertex, including moving the lines across each other, leaves the diagram unchanged. For example, the diagrams shown in Figure \ref{fig:equiv} are equivalent, while the diagrams shown in Figure \ref{fig:diff}  are different. In \cite{fdiag} all the Feynman diagrams with one external electron line and up to six vertices are explicitly shown.

  \begin{figure}
\centering
   \begin{center}
   \begin{tikzpicture}[node distance=1cm and 1cm]
\coordinate[] (v1);
\coordinate[below=1.05 cm of v1] (v2);
\coordinate[above=1.05 cm of v1] (v3);
\coordinate[right=1.05 cm of v1] (v4);
\coordinate[right= 1.05 cm of v4] (v5);
\coordinate[right= 0.55 cm of v4] (v6);
\coordinate[above= 0.55 cm of v6] (v7);
\coordinate[below= 0.55 cm of v6] (v8);
\draw[electron] (v2) -- (v3);
\draw[photonloop] (v1) -- (v4);
\draw[photonloop](v7) -- (v8);
\semiloop[electron]{v4}{v5}{0};
\semiloop[electron]{v5}{v4}{180};
\node[below=1.40 cm of v4] {\textit{(a)}};
\end{tikzpicture}
\hspace{2cm}   
  \begin{tikzpicture}[node distance=1cm and 1cm]
\coordinate[] (v1);
\coordinate[below=1.05 cm of v1] (v2);
\coordinate[above=1.05 cm of v1] (v3);
\coordinate[right=1.05 cm of v1] (v4);
\coordinate[right= 1.05 cm of v4] (v5);
\coordinate[right= 0.55 cm of v4] (v6);
\coordinate[above= 0.55 cm of v6] (v7);
\coordinate[below= 0.55 cm of v6] (v8);
\draw[electron] (v2) -- (v3);
\draw[photonloop] (v1) -- (v4);
\semiloop[electron]{v4}{v5}{0};
\semiloop[electron]{v5}{v4}{180};
\semiloopt[photonloop]{v7}{v8}{210};
\node[below=1.40 cm of v4] {\textit{(b)}};
\end{tikzpicture}
\caption{These two graphs represent the same Feynman diagram.} \label{fig:equiv}
\end{center}   
   \end{figure}
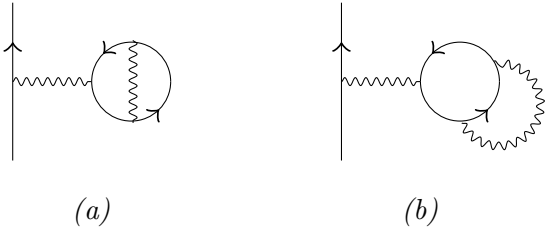

  \begin{figure}
\centering
   \begin{center}
   \begin{tikzpicture}[node distance=1cm and 1cm]
\coordinate[] (v1);
\coordinate[below=1.05 cm of v1] (v2);
\coordinate[below=0.80 cm of v1] (v9);
\coordinate[above=0.80 cm of v1] (v10);
\coordinate[above=1.05 cm of v1] (v3);
\coordinate[right=1.05 cm of v1] (v4);
\coordinate[right=0.30  cm  of v4] (v11);
\coordinate[above=0.60 cm of v11] (v13);
\coordinate[right=0.30  cm  of v4] (v12);
\coordinate[below=0.60  cm  of v12] (v14);
\coordinate[right= 1.60 cm of v4] (v5);
\coordinate[right= 0.80  cm of v4] (center);
\coordinate[above= 0.80 cm of center] (v7);
\coordinate[below= 0.80 cm of center] (v8);
\draw[photonloop] (v13) -- (v14);
\draw[electron] (v2) -- (v3);
\draw[photonloop] (v10) -- (v7);
\draw[photonloop] (v9) -- (v8);
\semiloopt[electron]{v4}{v5}{0};
\semiloopt[electron]{v5}{v4}{180};
\node[below=1.40 cm of v4] {\textit{(a)}};
\end{tikzpicture}
\hspace{2cm}   
  \begin{tikzpicture}[node distance=1cm and 1cm]

\coordinate[] (v1);
\coordinate[below=1.05 cm of v1] (v2);
\coordinate[below=0.80 cm of v1] (v9);
\coordinate[above=0.80 cm of v1] (v10);
\coordinate[above=1.05 cm of v1] (v3);
\coordinate[right=1.05 cm of v1] (v4);
\coordinate[right=1.30  cm  of v4] (v11);
\coordinate[above=0.60 cm of v11] (v13);
\coordinate[right=1.30  cm  of v4] (v12);
\coordinate[below=0.60  cm  of v12] (v14);
\coordinate[right= 1.60 cm of v4] (v5);
\coordinate[right= 0.80  cm of v4] (center);
\coordinate[above= 0.80 cm of center] (v7);
\coordinate[below= 0.80 cm of center] (v8);
\draw[photonloop] (v13) -- (v14);
\draw[electron] (v2) -- (v3);
\draw[photonloop] (v10) -- (v7);
\draw[photonloop] (v9) -- (v8);
\semiloopt[electron]{v4}{v5}{0};
\semiloopt[electron]{v5}{v4}{180};
\node[below=1.40 cm of v4] {\textit{(b)}};
\end{tikzpicture}
\caption{These two graphs represent different Feynman diagrams.} \label{fig:diff}
\end{center}   
   \end{figure}
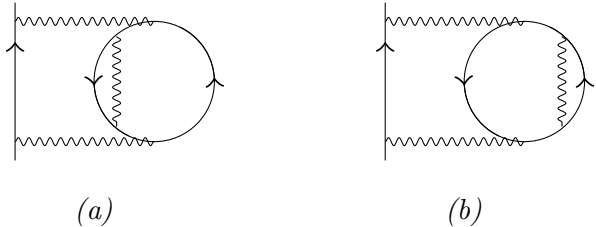

\vskip 0.2cm 

\noindent For the present work we will not need to consider Feynman diagrams with external photon lines, for reasons that will become clear  below.

\vskip 0.2cm

\noindent Photon and electron lines are connected to each other by the rules formalized by the Wick theorem which we now explain.  

\vskip 0.2cm

\noindent  In quantum field theory,  Feynman diagrams may  be generated using either path integrals or canonical quantization. We will use path integration in Section \ref{sect_QFT} to obtain generating functions for the numbers of  Feynman diagrams but for the present section, canonical quantization is more useful. 

\vskip 0.2cm

\noindent  In canonical quantization,  these fields  are operators acting on Fock space and obeying prescribed commutation relations, see for example \cite{weinberg} for details. The distinction between the operators  $\phi$ and $\phi^*$ provides the orientation of the corresponding lines in the Feynman diagram.  Each $(A \phi^* \phi)$ factor  corresponds to a vertex where a photon line terminates on an electron line. 

\vskip 0.2cm

\noindent  In this approach, Feynman diagrams are obtained by considering correlation functions which, for $N$ external electron lines and $v$ vertices, take the following form:
  \be
 \bigl{\langle} (\phi^*)^N  \bigl| ~  ( A_1 \phi_1 ^* \phi_1 ) \ldots ( A_v \phi_v^* \phi_v) ~\bigr| \phi^N  \bigr{\rangle } \,  \label{correl}
  \ee 
  where the $N$ fields $\phi$ in the ket $\bigr| \phi^N  \bigr{\rangle }$ represent $N$ electron fields  in the infinite past (in the infinite negative  direction along the corresponding electron lines) and the  $N$ fields $\phi^*$ in the bra $ \bigl{\langle} (\phi^*)^N  \bigl|$ represent $N$ electron fields  in the infinite future. These correlators correspond to Feynman diagrams with N external electron lines but we  also refer to them as $2N$-point electron correlation functions or, for short, $2N$-point functions because of the $2N$  electron fields appearing in the bra and ket.  Due to the type of interaction we are considering, the numbers of electrons in the bra and in the ket must be the same, otherwise the correlation function is trivially zero. 
 
 \vskip 0.2cm
 
\noindent   These correlation functions are associated to Feynman diagrams 
  as  dictated by  Wick's theorem, which amounts to the following rules:
\begin{itemize}

\item Each vertex $(A_i \phi_i^* \phi_i)$ is represented by a trivalent vertex with a solid outgoing half-edge corresponding to $\phi_i$, a solid incoming edge corresponding to $\phi_i^*$, and one non-oriented wavy edge corresponding to the field $A_i$;

\item One must consider all possible pairings of an $A$ field with another $A$ field in the correlator (\ref{correl}) and a $ \phi$ field with a $\phi^*$ field, leaving no field unpaired. The Feynman diagram is then obtained by the corresponding pairing of the half-edges of the trivalent vertices.
\end{itemize}  
  
  \vskip 0.2cm
  
\noindent Thus   each pairing  of $A$ fields  corresponds to a wavy line  in the Feynman diagram (or a photon  propagator, in the language of physics), and each pairing of $\phi $ and $\phi^*$ corresponds to an oriented solid line (an electron propagator). The pairing with the fields in the bra and ket correspond to external  lines, going to infinity. 

\vskip 0.2cm

\noindent   When all fields have been paired in the correlator (\ref{correl}) in some way, we will refer to the result as a Wick contraction. Two Wick contractions are identified if they differ only by a relabelling of the vertices $(A_i \phi_i^* \phi_i)$.
  
  \vskip 0.2cm
  
\noindent   Note that since there are no external photon fields, in other words there are no  A-fields in the bra and ket, the number of vertices has to be even for the correlation function to be non-vanishing and the number of photon lines (wavy edges) is then half the number of vertices, $\e = v/2$.
  
\vskip 0.2cm

\noindent    The Feynman rules of quantum field theory assign expressions to each Feynman diagram but these will not be needed here as we are only interested in drawing and counting Feynman diagrams.
  
  \vskip 0.2cm
  
\noindent  For example, for $N=1$ and $v=2$, we obtain the following four possibilities, illustrated in Figure \ref{fig:M1}  (there is no special meaning to drawing the pairings below or above, the two are used to make it easier to read the expressions):

  \ba   &~&
 (a)~~~ ~
    \contraction{ \bigl{\langle}}{\phi}{^*\bigl|~ (A_1 \phi_1^*}{\phi}
  \contraction[2ex]{  \bigl{\langle} \phi^*\bigl| ~  ( A_1   }{\phi}{_1^* \phi_1) ~ ( A_2 \phi_2^* \phi_2) ~ \bigr| }{\phi}
 \bcontraction{ \bigl{\langle} \phi^*\bigl|~(}{A}{_1 \phi_1^* \phi_1 )~(}{A}
 \bcontraction{ \bigl{\langle} \phi^*\bigl| ~  ( A_1 \phi_1^* \phi_1) ~  (A_2}{\phi}{_2^*}{\phi}
  \bigl{\langle} \phi^*\bigl| ~  ( A_1 \phi_1^* \phi_1) ~ ( A_2 \phi_2^* \phi_2)~ \bigr| \phi  \bigr{\rangle } 
  \,, ~~~~~~~(b)~~~~
  \contraction{ \bigl{\langle}}{\phi}{^*\bigl|~ (A_1 \phi_1^*}{\phi}
  \contraction[2ex]{  \bigl{\langle} \phi^*\bigl| ~  ( A_1 }{\phi}{_1^* \phi_1) ~  (A_2 \phi_2^*}{\phi}
 \bcontraction{ \bigl{\langle} \phi^*\bigl|~(}{A}{_1 \phi_1^* \phi_1 )~(}{A}
 \bcontraction{ \bigl{\langle} \phi^*\bigl| ~ (  A_1 \phi_1^* \phi_1) ~  (A_2}{\phi}{_2^* \phi_2 )~ \bigr|}{\phi}
  \bigl{\langle} \phi^*\bigl| ~  ( A_1 \phi_1^* \phi_1 )~  (A_2 \phi_2^* \phi_2) ~ \bigr| \phi  \bigr{\rangle } 
  \,,  \nonumber \\ \nonumber  \\ \nonumber   \\ \nonumber  \\ &~& (c)~~~~
     \contraction{ \bigl{\langle}}{\phi}{^*\bigl|~ (A_1 \phi_1^*   \phi_1 )~  (A_2 \phi_2^* \phi_2 )~ \bigr|   }{\phi}
  \contraction[2ex]{  \bigl{\langle} \phi^*\bigl| ~  ( A_1   }{\phi}{_1^* \phi_1) ~  (A_2 \phi_2^*       }{ \phi }
 \bcontraction{ \bigl{\langle} \phi^*\bigl|~(}{A}{_1\phi_1^* \phi_1) ~(}{A}
 \bcontraction[2ex]{ \bigl{\langle} \phi^*\bigl| ~  ( A_1 \phi_1^*}{\phi}{_1)~ ( A_2}{\phi}
  \bigl{\langle} \phi^*\bigl| ~   (A_1 \phi_1^* \phi_1) ~ ( A_2 \phi_2^* \phi_2 )~ \bigr| \phi  \bigr{\rangle } 
  , ~~~~~~~(d)~~~~  
    \contraction{ \bigl{\langle}}{\phi}{^*\bigl|~( A_1 \phi_1^*   \phi_1 )~  (A_2 \phi_2^* \phi_2) ~ \bigr|   }{\phi}
  \contraction[2ex]{  \bigl{\langle} \phi^*\bigl| ~  ( A_1   }{\phi}{_1^*}{ \phi }
 \bcontraction{ \bigl{\langle} \phi^*\bigl|~(}{A}{_1\phi_1^* \phi_1) ~(}{A}
 \bcontraction{ \bigl{\langle} \phi^*\bigl| ~  ( A_1 \phi_1^* \phi_1) ~  (A_2}{\phi}{_2^*}{\phi}
  \bigl{\langle} \phi^*\bigl| ~   (A_1 \phi_1^* \phi_1 )~  (A_2 \phi_2^* \phi_2 )~ \bigr| \phi  \bigr{\rangle } 
   \, .\label{wickcon}
 \ea

   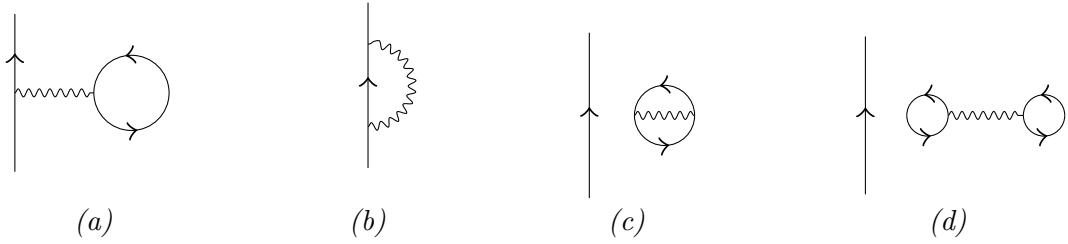
\begin{figure}
\centering
   \begin{center}
   \begin{tikzpicture}[node distance=1cm and 1cm]
%\coordinate[vertex] (v1);
\coordinate[] (v1);
\coordinate[below=1.05 cm of v1] (v2);
\coordinate[above=1.05 cm of v1] (v3);
\coordinate[right=1.05 cm of v1] (v4);
\coordinate[right=of v4] (v5);
\draw[electron] (v2) -- (v3);
\draw[photonloop] (v1) -- (v4);
\semiloop[fermion]{v4}{v5}{0};
\semiloop[fermion]{v5}{v4}{180};
\node[below=1.40 cm of v4] {\textit{(a)}};
\end{tikzpicture}
\hspace{2cm}   
   \begin{tikzpicture}[node distance=1cm and 1cm]
\coordinate[] (v1);
\coordinate[below=0.55 cm of v1] (v2);
\coordinate[above=0.55 cm of v1] (v3);
\coordinate[below=0.55 cm of v2] (v4);
\coordinate[above=0.55 cm of v3] (v5);
\semiloop[photonloop]{v3}{v2}{270};
\draw[electronmiddle] (v4) -- (v5);
\node[below= 0.40 cm of v4] {\textit{(b)}};
\end{tikzpicture}
\hspace{2cm}
 \begin{tikzpicture}[node distance=1cm and 1cm]
\coordinate[] (v1);
\coordinate[below=0.55 cm of v1] (v2);
\coordinate[above=0.55 cm of v1] (v3);
\coordinate[below=0.55 cm of v2] (v4);
\coordinate[above=0.55 cm of v3] (v5);
\coordinate[right=0.50 cm of v1] (v8);
\coordinate[right=0.60 cm of v1] (v6);
\coordinate[right= 0.80 cm of v6] (v7);
\semiloop[electronmiddle]{v6}{v7}{0};
\semiloop[electronmiddle]{v7}{v6}{180};
\draw[photonloop] (v6)--(v7);
\draw[electronmiddle] (v4) -- (v5);
\node[below=1.10 cm of v8] {\textit{(c)}};
\end{tikzpicture}
\hspace{2cm}
   \begin{tikzpicture}[node distance=1cm and 1cm]
\coordinate[] (v1);
\coordinate[below=1.05 cm of v1] (v2);
\coordinate[above=1.05 cm of v1] (v3);
\coordinate[right=0.55 cm of v1] (v4);
\coordinate[right= 0.55 cm of v4] (v5);
\coordinate[right= 0.55 cm of v4] (v8);
\coordinate[right=of v5] (v6);
\coordinate[right= 0.55 cm of v6] (v7);
\draw[electronmiddle] (v2) -- (v3);
\draw[photonloop] (v5) -- (v6);
\semiloop[fermion]{v4}{v5}{0};
\semiloop[fermion]{v5}{v4}{180};
\semiloop[fermion]{v6}{v7}{0};
\semiloop[fermion]{v7}{v6}{180};
\node[below=1.10 cm of v8] {\textit{(d)}};
\end{tikzpicture}
\caption{The four Feynman diagrams corresponding to the Wick contractions of Eq.(\ref{wickcon}).} \label{fig:M1}
\end{center}   
   \end{figure}
   
   \bigskip
   \bigskip

\noindent  The following two Wick  contractions are not distinct because they differ by a relabelling  of the two vertices:
 
 \be    \nonumber\contraction{ \bigl{\langle}}{\phi}{^* \bigl|~ (A_1 \phi_1^*}{\phi}
  \contraction[2ex]{  \bigl{\langle} \phi^* \bigl| ~  ( A_1   }{\phi}{_1^* \phi_1) ~ ( A_2 \phi_2^* \phi_2 )~ \bigr| }{\phi}
 \bcontraction{ \bigl{\langle} \phi^* \bigl|~(}{A}{_1 \phi_1^* \phi_1 )~(}{A}
 \bcontraction{ \bigl{\langle} \phi^* \bigl| ~  ( A_1 \phi_1^* \phi_1) ~  (A_2}{\phi}{_2^*}{\phi}
  \bigl{\langle} \phi^* \bigl| ~  ( A_1 \phi_1^* \phi_1) ~ ( A_2 \phi_2^* \phi_2) ~ \bigr| \phi  \bigr{\rangle } ~~~~~
  \text{and} ~~~~~
     \contraction{ \bigl{\langle}}{\phi}{^* \bigl|~ (A_1 \phi_1^*\phi_1 ) ~ ( A_2 \phi_2^*}{\phi}
  \contraction[2ex]{  \bigl{\langle} \phi^* \bigl| ~  ( A_1  \phi_1^* \phi_1) ~ ( A_2 }{\phi}{_2^* \phi_2) ~ \bigr| }{\phi}
 \bcontraction{ \bigl{\langle} \phi^* \bigl|~(}{A}{_1 \phi_1^* \phi_1 )~(}{A}
  \bcontraction[2ex]{ \bigl{\langle} \phi^*| ~  ( A_1}{ \phi}{_1^*}{\phi}
 \bigl{\langle} \phi^* \bigl| ~  ( A_1 \phi_1^* \phi_1) ~ ( A_2 \phi_2^* \phi_2) ~ \bigr| \phi  \bigr{\rangle } \, . ~~~~~
  \ee
\noindent Note that this procedure generates both connected and disconnected diagrams. We are only interested in connected diagrams which means that we consider pairings for which no subsets of fields are paired only amongst themselves.

\section{Correspondence between Feynman diagrams and $N$-rooted ribbon graphs}
\label{sect_bijection}

In this section we establish the bijective correspondence between  the $2N$-point electron correlation functions
 in scalar  QED and $N$-rooted ribbon graphs. As explained in the previous sections, Feynman diagrams in quantum field theory are generated by Wick's theorem, while rooted ribbon graphs can be realized as pairs of permutations on the set $H$  of half-edges of the ribbon graph. Our strategy will be to associate to each Wick contraction leading to a connected $2N$-point Feynman diagram with $\e$ photon lines  a pair of permutations $(\alpha, \sigma)$ on the set of $2 \e$ half-edges that establishes the correspondence with a rooted ribbon graph.

\vskip 0.2cm

\noindent Before we prove the exact bijection, let us try to understand intuitively why such a bijection can be anticipated. Both Feynman diagrams and rooted ribbon graphs can be generated as combinatorial objects out of sub-structures (vertices, edges, propagators, etc.) with set of rules giving the relation between the sub-structures. 

\vskip 0.2cm

\noindent Let us, for the sake of simplicity, focus on the case of the two point function and its correspondence to one-rooted ribbon graphs. The generalization to $2N$-point functions is straightforward. In the case of Feynman diagrams we have two kinds of lines, the photon and electron lines, and one kind of vertex where a photon line ends on an electron line. The number of vertices occurring in a Feynman diagram with no external photon lines is already determined and given by twice the number of internal photon lines in the diagram, leaving only two combinatorial objects to consider and a rule for how they combine. In addition, in the case of the $2$-point function, there is one special electron line which extends from $-\infty$ to $\infty$ along the vertical time axis, while all the other electron lines are closed loops. In the case of one-rooted ribbon graphs, we again have two combinatorial objects, vertices and edges, which combine in a given way. In addition, there is again one specific vertex which is special -- the root vertex. Thus, combinatorially, both structures are  determined by the exact same amount of combinatorial data and we have just the right amount of it to expect a correspondence.  Finally, previous results on enumeration of maps and Feynman diagrams show, see \cite{fdiag}, that the number of $2$-point functions arranged by the number of photon lines and the number of one-rooted ribbon graphs arranged by the number of edges are exactly the same. This implies one can make such a correspondence concrete by making the right identifications on each side.

\vskip 0.2cm

\noindent To understand which structures relate to which, we note that the simplest object we can consider for a $2$-point function is just an external electron propagator with no photon lines. On the other hand, in the case of one-rooted ribbon graphs the simplest  object is a single vertex and no edges. Thus, the correspondence should relate electron propagators to vertices in rooted ribbon graphs implying photon lines get mapped to edges in a rooted ribbon graph. In addition, the correspondence of the external line to  the root vertex in the ribbon graph generalizes for the case of  $2N$-point functions to the relation between the $N$ external electron lines and the $N$ root  vertices in the $N$-rooted ribbon graph. With this intuitive idea, we can hope to establish the exact bijection between the two sides.

\subsection{Bijection between Feynman diagrams and $N$-rooted ribbon graphs}
\label{sect_bijection_proof}

In this subsection we make precise the correspondence between Feynman diagrams of scalar QED and $N$-rooted ribbon graphs.

\vskip 0.2cm

\noindent  To get a handle on the number of Feynman diagrams in the theory, we use the fact that all Feynman diagrams are generated by Wick's theorem and use this to define a bijection between Feynman diagrams for $2N$-point functions and $N$-rooted graphs. As explained previously, the contractions from Wick's theorem generate both connected and disconnected Feynman diagrams. We consider only contractions leading to connected Feynman diagrams. 
 
\begin{theorem}
\label{thm_bijection}
There is a one-to-one correspondence between the set of connected Feynman diagrams of  scalar QED with $N$ external electron lines and  $\e$ internal photon lines on one side and the set of $N$-rooted maps with $\e$ edges on the other. 
\end{theorem}

\begin{example}

Under the bijection from Theorem \ref{thm_bijection} the Feynman diagram in Figure \ref{fig:perms} corresponds to the ribbon graph defined by the set of half-edges $H=\{1, 2, \dots, 12\}$ and the permutations $\alpha=(1\,2)(3\,4)(5\,6)(7\,8)(9\,10)(11\,12)$ and $\sigma=(\hat{1})(\hat{2}\,7\,3\,9) (\hat{11}\,10\,8\,12) (4\,6\,5)$ where the first three cycles of the permutation $\sigma$ are labelled with the numbers from one to three, respectively, and the hat symbol marks the root half-edges. 

\begin{figure}
\centering
   \begin{center}
   \begin{tikzpicture}[node distance=1cm and 1cm]
\coordinate[] (v1);
\coordinate[below=1.55 cm of v1, label=right:$\,_1$] (v2);

\coordinate[above=1.55 cm of v1] (v3);
\coordinate[right=1.05 cm of v1] (v4);

\coordinate[right= 1.60 cm of v4] (v5);
\path (v5)+(350:0.05) coordinate[ label=below:$\,_4$];

\coordinate[right= 0.80  cm of v4] (center);
\coordinate[above= 0.80  cm of center] (top);
\coordinate[below= 0.80  cm of center] (bottom);
\path (bottom)+(110:0.55) coordinate[ label=below:$\,_5$];
\path (top)+(190:0.2) coordinate[ label=below:$\,_6$];
\draw[photonloop] (top) -- (bottom);
\draw[electronmiddle] (v2) -- (v3);
\coordinate[below=1.20 cm of v1] (v6);
\coordinate[right=3.50 cm of v6] (v22);
\semiloopt[electronf]{v4}{v5}{0};
\draw[photonloop] (v6) -- (v22);
\coordinate[below=0.35 cm of v22, label=left:$\,_2$] (v23);
\coordinate[above=3.10 cm of v23] (v24) ; 
\coordinate[right=3.50 cm of v23] (v33);
\coordinate[above=3.10 cm of v33] (v34) ;
\coordinate[above=1.55 cm of v33] (v31);
\coordinate[above=0.75 cm of v31](v35);

\coordinate[below=0.75 cm of v31](v36);
\coordinate[above=1.05 cm of v31](v37);
\path (v37)+(70:0.5) coordinate[ label=below:$\,_{12}$];
\coordinate[below=1.05 cm of v31](v38);
\coordinate[above=1.55 cm of v23](v21);
\path (v21)+(180:0.15) coordinate[ label=below:$\,_{3}$];
\coordinate[right=1.55 cm of v21](v39);
\coordinate[left=0.10 cm of v39](v40);
\coordinate[below=0.10 cm of v40](v41);
\coordinate[right=0.45 cm of v39](v42);
\coordinate[above=0.10 cm of v42](v43);
\draw[electronmiddle](v33)--(v34);
\draw[electronfff] (v23) -- (v24);

\draw[photonloop](v21) -- (v5);
\coordinate[below=0.75 cm of v21](v25);
\path (v25)+(15:0.1) coordinate[ label=below:$\,_7$];
\coordinate[above=0.75 cm of v21](v26);
\path (v26)+(60:0.4) coordinate[ label=below:$\,_{9}$];
\draw[photonloop](v25)--(v41);
\draw[photonloop](v43)--(v35);
\path (v35)+(120:0.4) coordinate[ label=below:$\,_{8}$];
\draw[photonloop](v26)--(v36);
\path (v36)+(160:0.4) coordinate[ label=below:$\,_{10}$];
\semiloopt[electronf]{v5}{v4}{180};
\semiloop[photonloop]{v37}{v38}{270};
\path (v38)+(0:0.2) coordinate[ label=below:$\,_{11}$];
\end{tikzpicture}
\caption{Feynman diagram corresponding to the permutations $\alpha=(1\,2)(3\,4)(5\,6)(7\,8)(9\,10)(11\,12)$, $\sigma=(\hat{1})(\hat{2}\,7\,3\,9) (\hat{11}\,10\,8\,12) (4\,6\,5)$ on the set of half-edges $H=\{1, 2, \dots, 12\}$. A different choice of labelling of half edges and the resulting permutations define an equivalent ribbon graph.} \label{fig:perms}
\end{center}   
   \end{figure}
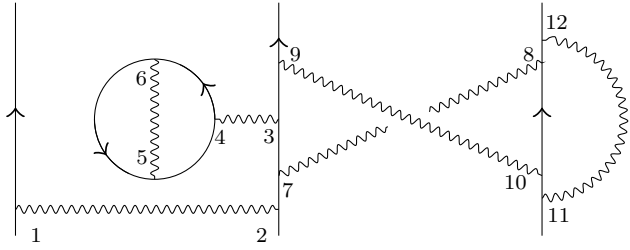

\end{example}

\begin{proof}
In general, a Wick contraction will contain $N$ strings of pairings linking each of  the fields  $\phi^*$ appearing in the bra to one and only one of the fields $\phi$ appearing in the ket. We label the pairs of fields thus linked by the same index:

\begin{equation*}
\contraction{\langle \phi^*_1 \ldots}{\phi}{^*_k \ldots \phi^*_N \mid (A_1\phi^*_1}{\phi}
\bcontraction{ \langle \phi^*_1 \ldots  \phi^*_k \ldots \phi^*_N \mid (A_1}{\phi}{^*_1\phi_1)(A_2\phi^*_2}{\phi)}
\contraction{\langle \phi^*_1 \ldots \phi^*_k \ldots  \phi^*_N \mid (A_1 \phi^*_1 \phi_1) (A_2}{\phi}{^*_2\phi_2) \ldots}{}
\contraction{\langle \phi^*_1 \ldots \phi^*_k \ldots  \phi^*_N \mid  (A_1 \phi^*_1 \phi_1) (A_2\phi^*_2\phi_2)\ldots \ldots (A_s }{\phi}{^*_s\phi_s)\ldots 
(A_{2\e}  \, \phi^*_{2\e}  \, \phi_{2\e}) \mid \phi_1 \ldots}{\phi}
\bcontraction{\langle \phi^*_1 \ldots \phi^*_k \ldots  \phi^*_N \mid  (A_1 \phi^*_1 \phi_1) (A_2\phi^*_2\phi_2)\ldots}{}{\ldots (A_s \phi^*_s}{\phi}
\langle \phi^*_1 \ldots \phi^*_k \ldots  \phi^*_N \mid  (A_1 \phi^*_1 \phi_1) (A_2\phi^*_2\phi_2)\ldots 
\ldots (A_s \phi^*_s\phi_s) \ldots 
(A_{2\e}  \, \phi^*_{2\e}  \, \phi_{2\e}) \mid \phi_1 \ldots \phi_k \ldots \phi_N \rangle 
\end{equation*}
where for later convenience we have labeled the $N$ external electron fields in the bra and ket with integers $1, \ldots, N$, and labeled the set of vertices in the diagram with integers $1, \ldots, 2\e$ appearing as subscripts of the corresponding fields. \\

\noindent Given such a Feynman diagram, we construct an $N$-rooted ribbon graph $(H, \alpha, \sigma)$ as follows. First, the set $H$ of half-edges is given by the set of vertices in the Wick contraction: define $H=\{1, \dots, 2\e\}$. 
Next we define the following two permutations based on the two kinds of pairings that appear in the corresponding Wick contraction.

\noindent For each photon propagator obtained by pairing the fields $A_{i}$ and $A_{j}$ between the $i^\textsuperscript{th}$ and $j^\textsuperscript{th}$ vertices with $i\neq j$, define an involution $\alpha_{ij} = (ij)$. 
Since there is just one $A$ field in each vertex, the pairing between the $A_{i}$ and $A_{j}$ fields determines the involution $\alpha_{ij}$ uniquely. In a Feynman diagram with $2\e$ vertices of the kind $(A \phi^* \phi)$, there will be $\e$ such pairings and hence $\e$ transpositions $\alpha_{ij}$. These transpositions form a permutation $\alpha\in S_{2\e}$ on the set $H$ of half-edges. It is a fixed point free involution by construction. \\

\noindent Now we follow the electron pairings. There are two kinds of electron pairings -- pairings involving the electrons in the bra and ket leading to external electron lines in the Feynman diagram and pairings that lead to internal electron loops in the Feynman diagram. Let us consider each separately. 
\begin{itemize}
\item In the former case, the line starts with the pairing of the field $\langle \ldots \phi^*_{k} \ldots  \mid$ with a field, say $\phi_p$, in the $p^\textsuperscript{th}$ vertex.  We then follow the sequence formed by the pairing of $\phi^*_p$ with the next field, say $\phi_q$, and so on until the sequence leads us to a pairing of some $\phi_s^*$ to the electron $\mid  \ldots \phi_{k} \ldots \rangle$ in the ket, i.e.
\begin{equation*}
\contraction{\langle\phi^*_1 \ldots }{\phi}{^*_k  \ldots \phi^*_N \mid \ldots (A_p\phi^*_p}{\phi}
\bcontraction{\langle \phi^*_1 \ldots \phi^*_k \ldots \phi^*_N \mid  \ldots  (A_p }{\phi}{^*_p \phi_p) \ldots  (A_q\phi^*_q}{\phi}
\contraction{\langle \phi^*_1 \ldots \phi^*_k \ldots \phi^*_N \mid  \ldots  (A_p \phi^*_p \phi_p) \ldots  (A_q\phi^*_q\phi_q)\ldots \ldots (A_s }{\phi}{^*_s\phi_s) \ldots \mid \phi_1 \ldots }{\phi}
\contraction{\langle \phi^*_1 \ldots \phi^*_k \ldots \phi^*_N \mid  \ldots  (A_p \phi^*_p \phi_p) \ldots  (A_q}{\phi}{^*_q\phi_q)\ldots}{}
\bcontraction{\langle \phi^*_1 \ldots \phi^*_k \ldots \phi^*_N \mid  \ldots  (A_p \phi^*_p \phi_p) \ldots  (A_q\phi^*_q\phi_q)\ldots}{}{\ldots (A_s \phi^*_s}{\phi}
\langle \phi^*_1 \ldots \phi^*_k \ldots \phi^*_N \mid  \ldots  (A_p \phi^*_p \phi_p) \ldots  (A_q\phi^*_q\phi_q)\ldots \ldots (A_s \phi^*_s\phi_s) \ldots \mid \phi_1 \ldots \phi_k \ldots \phi_N \rangle .
\end{equation*}
To the external electron line formed by this sequence we associate the cycle $(p,q,\ldots, s)$ and mark the half-edge $p \in H$ corresponding to the first vertex  $(A_p \phi^*_p \phi_p)$ linked to $\phi^*_k$. The half-edge $p$ becomes thus the $k^\textsuperscript{th}$ root and we denote it by  $\hat h_k$. 
There is only one $\phi^*$ and $\phi$ field in each vertex $(A\phi^*\phi)$ and hence the cycle permutation associated to each sequence of pairings is necessarily unique. Since there are $N$ electrons in the bra or ket, there are $N$ such cycles with $N$ labeled half-edges, $\hat h_1, \ldots, \hat h_N$.
\item In the latter case,  the sequence of pairings both starts and ends at the same vertex leading to an internal electron loop, e.g.
\begin{equation*}
\bcontraction{\langle  \phi^*_1 \ldots \phi^*_k \ldots \phi^*_N\mid \ldots (A_k}{\phi}{^*_k \phi_k) \ldots (A_l\phi^*_l}{\phi)}
\contraction[2ex]{\langle  \phi^*_1 \ldots \phi^*_k \ldots \phi^*_N\mid \ldots (A_k \phi^*_k}{\phi}{_k) \ldots (A_l\phi^*_l\phi_l)\ldots \ldots (A_n}{\phi}
\bcontraction{\langle  \phi^*_1 \ldots \phi^*_k \ldots \phi^*_N \mid \ldots (A_k \phi^*_k \phi_k) \ldots (A_l\phi^*_l\phi_l)\ldots}{}{\ldots (A_n \phi^*_n}{\phi}
\contraction{\langle  \phi^*_1 \ldots \phi^*_k \ldots \phi^*_N  \mid \ldots (A_k \phi^*_k \phi_k) \ldots (A_l}{\phi}{^*_l\phi_l)\ldots}{}
\langle  \phi^*_1 \ldots \phi^*_k \ldots \phi^*_N  \mid \ldots (A_k \phi^*_k \phi_k) \ldots (A_l\phi^*_l\phi_l)\ldots \ldots (A_n \phi^*_n\phi_n) \ldots \mid  \phi_1 \ldots \phi_k \ldots \phi_N \rangle .
\end{equation*}
To this internal electron loop we associate the cycle $(k,l,\ldots,n)$. The same argument as before implies the cycle associated to the sequence of pairings is also unique. 
\end{itemize}

\noindent Since none of the fields must be left unpaired, the set of all such cycles gives a permutation $\sigma$ on the set of $2e$ vertices, that is on the set $H$. 

\vskip 0.2cm

\noindent Since we only consider connected Feynman diagrams, the obtained permutations $\alpha$ and $\sigma$ generate a subgroup of $S_{2\e}$ that acts transitively on $H$. This completes the construction of a ribbon graph $(H, \alpha, \sigma)$ with $N$ roots $\hat h_1, \ldots, \hat h_N$ corresponding to a given Wick contraction. \\

\vskip 0.2cm

\noindent Conversely, let $(H, \alpha, \sigma)$ be an $N$-rooted ribbon graph with $\e$ edges and roots $\hat h_1, \ldots, \hat h_N$. The set $H=\{1,\dots, 2\e\}$ of half edges determines the set of vertices in a Wick contraction. The number $N$ of roots determines the number of fields $\phi^*$ and $\phi$ in the bra and ket. The permutation $\alpha\in S_{2\e}$ gives the pairings of the $A$-fields. The permutation $\sigma\in S_{2\e}$ has $N$ special cycles each of which contains a root $\hat h_i$ for $i=1,\dots, N$. Let the cycle containing  the root $\hat h_k$ have the form $(\hat h_k, p, q, \dots, s)$. This cycle determines a string of pairings that connects $\phi^*_k$ to $\phi_p$, followed by pairing $\phi^*_p$ to $\phi_q$, and so on up to the pairing of $\phi^*_s$ to $\phi_k$ in the ket.  The remaining cycles of $\sigma$ determine the cycle pairings of the fields $\phi^*$ and $\phi$ in the obvious way.  

\vskip 0.2cm

\noindent To see that equivalent Wick contractions  arise from equivalent rooted ribbon graphs, recall that two Wick contractions are identified if and only if they can be obtained from one another by relabelling of vertices.  Denote such a relabelling by  $r\in S_{2\e}$. By our construction, vertices in a Wick contraction form the set $H$ of half-edges in the corresponding ribbon graph. Therefore a relabelling of vertices in a contraction results in a relabelling of the half-edges,  $r:H\to H$,  and the new permutations $(\alpha', \sigma')$ are related to the original permutations $(\alpha, \sigma)$ as follows: 
\begin{equation*}
\alpha' = r\circ \alpha \circ r^{-1}, \qquad \sigma' = r\circ \sigma \circ r^{-1}\,.
\end{equation*}
\noindent We want to show that the two ribbon graphs $(H, \alpha, \sigma)$ and $(H, \alpha', \sigma')$ are isomorphic. 
The relabelling permutation $r$ is the bijection $\psi:H\to H$ from Definition \ref{def_ribbon_iso} and thus the required commutation relations  are satisfied. 

\vskip 0.2cm

\noindent The roots are mapped to the roots by the relabelling $r$ as the $k^\textsuperscript{th}$ root half-edge corresponds to the vertex of the Wick contraction which is paired to $\phi^*_k$ in the bra-part. Since the relabelling does not affect the fields in the bra and ket, we get that $r$ sends $k^\textsuperscript{th}$ root of the graph $(H, \alpha, \sigma)$ to the $k^\textsuperscript{th}$ root of the graph $(H, \alpha', \sigma')$.

\end{proof}

\begin{remark} Let $V$ denote the set of cycles of the permutation $\sigma$. 
Under the bijection of Theorem \ref{thm_bijection} between $2N$-point Feynman diagrams in scalar QED and $N$-rooted graphs, the set of $N$ external electron lines, the set of $(| V|-N)$ internal electron loops, the $2\e$ vertices, and the $e$ photon lines  in the Feynman diagram correspond bijectively to the $N$ root vertices, the $(|V|-N)$ non-root vertices, the $2\e$ half-edges and the $\e$ edges of the $N$-rooted  graph, respectively. Note that in physics, it is important to arrange the Feynman diagrams by the number of loops, where the loops include the internal electron and photon loops. As is easy to see, the number of such loops is $e-N+1$. Thus the number of Feynman diagrams with $l$ loops is equal to the number of $N$-rooted graphs with $l+N-1$ edges. 

\end{remark}

  \section{Fundamentals of quantum field theory}
  \label{sect_QFT}
  
\noindent   From the point of view of quantum field theory, Feynman diagrams can be generated through path integrals with respect to the  fields of the theory, which are functions or sections of fiber bundles over spacetime, in general  a d-dimensional manifold. A quantum field theory that generates the Feynman diagrams of interest to this work is the theory of a complex scalar field $\phi$ coupled to a real scalar field $A$, for which one defines the partition function
  \be Z \left(J,\eta,\eta^*, \lambda, \hbar \right)
  = \int {\cal{D}}\phi {\cal{D}}\phi^* {\cal{D}}A~~ \exp \bigl( \frac{{\cal{S}}}{\hbar} \bigr)  \, ,\label{part}
  \ee
  where  $\hbar$ is Planck's constant, $\lambda$ is the coupling constant of the theory to be introduced below and the notation ${\cal{D}} \phi$ denotes the measure used to define a path integral with respect to the field $\phi$. For the purpose of this
  paper, there will be no need to specify this definition  as we will be considering a simplified situation in which the path integration will reduce to an ordinary integration.
  
\noindent   The action ${\cal{S}}$ may be separated into three contributions,   \be \nonumber{\cal{S}} = {\cal{S}}_F + {\cal{S}}_I + {\cal{S}}_S \, . \ee
  The first term represents the free field action (containing the kinetic and rest mass energy contributions but no interactions),
  \be \nonumber{\cal{S}}_F =  - \int d^dx ~\Bigl(  m_\phi \, \phi(x) \phi^*(x)  +\frac{m_A}{2} A^2(x)  +  \partial_\mu \phi(x) \partial^\mu \phi^* (x) + 
  \frac{1}{2} \partial^\mu A(x) \partial^\mu A(x)  \Bigr) .
  \ee In the following we will set the two masses equal, $m_A = m_\phi$, and will choose units in which the masses are equal to 1. 
  The integral is over the d dimensions of the spacetime, with coordinates $x_1, \ldots, x_d$. It is not necessary to provide any more details about this integral because for the purpose of the present work, this integration will be absent for reasons discussed in the next section.
  The interaction action relevant to our calculation is 
  \be \nonumber
   {\cal{S}}_I = \lambda \, \int d^d x~  \phi^*(x) \phi(x) A (x) \, ,
  \ee where $\lambda$ is the coupling constant of the theory and is used to construct the perturbative expansion of the integrals. 
   The last term contains  the so-called source terms 
   \be\nonumber
   {\cal{S}}_S = \int d^dx ~ \Bigl(J(x) A(x) +  \eta^*(x) \phi(x) \, ,
  + \eta(x) \phi^*(x)  \Bigr),
  \ee where the sources $J(x)$,  $\eta(x)$ and $\eta^*(x)$ are  formal parameters used to  compute the correlation functions of the theory as described below.
  
  \vskip 0.2cm
  
  \noindent  In the following we use the convention of particle physicists who choose units with $\hbar=1$. 
  
  \vskip 0.2cm
  
\noindent   The correlation functions of the theory are the quantities of interest in physics.   An example would be  the four-point  correlation function $\langle \phi \phi^* A A\rangle $, where we have suppressed the spacetime dependence of all the fields for ease of notation. This correlation function is  given by
  \be
\langle \phi \phi^* A^2 \rangle := \frac{1}{Z(0)}   \int {\cal{D}}\phi {\cal{D}}\phi^* {\cal{D}}A~~\Bigl( \phi \phi^*A^2  ~ e^{{\cal{S}}_F + {\cal{S}}_I}
\Bigr) \,  , \label{nos}
  \ee where
  \be  Z(0) :=  Z \left(J=\eta=\eta^*=\lambda =0,\hbar \right) \, . \label{zzz} \ee  Note the absence of the sources in Eq.(\ref{nos}).
  The usefulness of the sources comes from the fact that   this   correlation function can be written as
  \be \langle \phi \phi^* A^2 \rangle  = \frac{1}{Z(0)}   \frac{\delta^4 Z \left(J,\eta,\eta^*,\hbar\right) }{ \delta J^2 \,  \delta \eta \, \delta \eta^*}  \Biggr|_{J=\eta=\eta^*=0}   \, ,\label{phiphi}
  \ee
  where the derivatives with respect to the sources  are actually functional derivatives (we have suppressed the spacetime dependence of the sources as well).
  
  \vskip 0.2cm

\noindent   Any correlation function can be obtained by taking such derivatives of the partition function (\ref{part})  and then setting all the sources equal to zero.   The correlation functions contain all the dynamics of the theory. Therefore, if one can explicitly calculate the partition function (\ref{part})  one has solved the theory in the sense that all correlation functions can then easily be  obtained.

  \section{Counting Feynman diagrams using quantum field theory}
  \label{sect_count}
  
\noindent   In practice, it is impossible to evaluate  the partition function (\ref{part})  except for extremely simple models. 
Fortunately, we are not interested here in solving the quantum theory, only in counting  Feynman diagrams. Then one can simplify greatly
the problem by considering our quantum field theory in zero spacetime dimension. What this means in practice is that our fields are now taken to be spacetime independent and the action does not contain an integral over spacetime anymore. From the point of view of the path integral, the fields become now ordinary real or complex variables and the  path integrals with respect to the fields reduce to ordinary integrals over $\mathbb R^3$.  
More precisely,  $A$ and $J$ are now real variables, and $\phi, \phi^*\in \mathbb C$ are complex conjugated to each other, which implies the same relationship for the sources $\eta, \eta^*\in\mathbb C$. Assuming $\phi=x+iy$ with $x,y\in \mathbb R$ we then regard $d\phi d\phi^*$  as  $dx dy$. 
 We therefore consider from now on the following integral of a real function

  \be Z \left(J,\eta,\eta^* , \lambda\right) 
  = \int_{\mathbb R^3} d\phi d\phi^* dA~~ \exp  \Biggl( ~- \phi \phi^* - \frac{1}{2} A^2 + \lambda \phi \phi^* A + J A +  \eta \phi^*
  + \eta^* \phi \Biggr) \, . \label{part2}
  \ee 
This is still nontrivial to evaluate, the culprit being the interaction term $\lambda \phi \phi^* A$.  A standard technique in QFT is to Taylor expand the exponential of  the interaction term  and  treat the result as a formal series which can be integrated  term by term. This results in an infinite number of integrals, forming a series in powers of the coupling constant $\lambda$,  each of which may be evaluated exactly.  The question of the convergence of the resulting series is a subtle issue in general but  is not relevant here as $\lambda$ will be used  as formal parameter that will allow us to count separately certain classes of diagrams.

\vskip 0.2cm

\noindent Expanding the exponential of the interaction term, we get
\be \nonumber
Z \left(J,\eta,\eta^*, \lambda \right)
  = \int_{\mathbb R^3} d\phi d\phi^* dA~~ \sum_{k=0} \frac{(\lambda \phi \phi^* A)^k}{k!} ~ \exp  \Biggl(~- \phi \phi^* - \frac{1}{2} A^2  + J A +  \eta \phi^*
  + \eta^* \phi \Biggr),
  \ee which, using the trick of obtaining factors in front of the exponential by taking derivatives with respect to the sources as  in Eq.(\ref{phiphi}), may be written as

\be \nonumber
Z \left(J,\eta_1,\eta_2, \lambda \right)  = \int_{\mathbb R^3} d\phi d\phi^* dA~~ \sum_{k=0} \frac{\lambda^k}{k!}   \left(\frac{d^3}{d \eta d \eta^*  d J}\right)^k ~ \exp  \Biggl(~- \phi \phi^* - \frac{1}{2} A^2  + J A +  \eta \phi^*
  + \eta^* \phi \Biggr)  \, .
\ee The derivatives with respect to the sources may be pulled out of the integrals which can  then be easily performed to give us

\ba \label{zees} Z \left(J,\eta,\eta^*, \lambda \right)   &=& \sum_{k=0} \frac{\lambda^k}{k!}  \left(\frac{d^3}{d\eta d\eta^* dJ}\right)^k~~  \int_{\mathbb R^3} d\phi d\phi^* dA~  ~ \exp  \Biggl( ~- \phi \phi^* - \frac{1}{2} A^2  + J A +  \eta \phi^*
  + \eta^* \phi \Biggr)  \nonumber \\
  &=& \pi  \sqrt{ 2 \pi} \, \sum_{k=0} \frac{\lambda^k}{k!}   \left(\frac{d^3}{d \eta d\eta^* dJ}\right)^k ~ \exp \left( \eta \eta^* + \frac{J^2}{2} \right)  \, .  \label{partfunc}
\ea
Note that $Z(0)$ as defined in Eq.(\ref{zzz}) is equal to
\be Z(0) =  \pi \sqrt{2 \pi} \,. \nonumber \ee

\noindent We may  now  obtain generating functions for the Feynman diagrams of relevance to the present work.  As follows from quantum field theory,  the generating function for the diagrams with  $p$ external photon lines and $N$ external electron lines is  given by 

\ba  Z_{N,p} \left( \lambda \right) &:=&  \langle (\phi \phi^*)^N \, A^p \rangle  \nonumber \\ 
&=&  \frac{1}{Z(0)} ~ \left(\frac{d}{d \eta d \eta^*} \right)^N \left( \frac{d}{dJ} \right)^p  \, Z \left(J,\eta,\eta^*, \lambda \right) \Biggr|_{J= \eta=\eta^*= 0 } \nonumber \\
&=& \left(\frac{d}{d \eta d \eta^*} \right)^N \left( \frac{d}{dJ} \right)^p 
 \sum_{k=0} \frac{\lambda^k}{k!}   \left(\frac{d^3}{d \eta d \eta^* d J}\right)^k ~ \exp \left( \eta \eta^* + \frac{J^2}{2} \right) 
\Biggr|_{J= \eta=\eta^*= 0 }  \, . \nonumber 
\ea
This is a generating function in the following sense: in the expansion of $Z_{N,p}(\lambda)$, the coefficient of  $\lambda^{k}$ gives the number of Feynman diagrams with $ p$ external photon lines, $N$ external electron lines, and $k$ vertices.  Figure (\ref{fig:perms}) shows one of the Feynman diagram corresponding to $p=0,N=3,k=12$. Although we will not work with Feynman diagrams with external photon lines, for illustrative purposes examples contributing to  $p=1, N=1, k=5$ are shown in Figure (\ref{fig:npk}).  Note that both connected and disconnected diagrams are counted by these generating functions.

  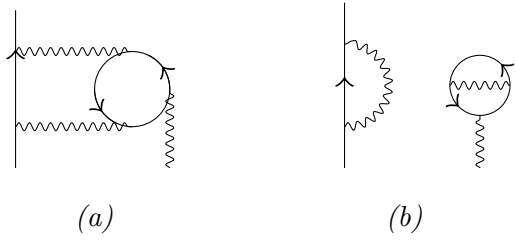
\begin{figure}
\centering
   \begin{center}
   \begin{tikzpicture}[node distance=1cm and 1cm]
\coordinate[] (v1);
\coordinate[below=1.05 cm of v1] (v2);
\coordinate[right= 2.05 cm of v2](v10);
\coordinate[above=1.05 cm of v1] (v3);
\coordinate[right=1.05 cm of v1] (v4);
\coordinate[right=of v4] (v5);
\coordinate[below= 0.50 cm of v1] (v6);
\coordinate[above= 0.50 cm of v1] (v7);
\coordinate[right=0.50 cm of v4] (center);
\draw[photonloop] (v10) -- (v5);
\coordinate[above= 0.5 cm of center] (v8);
\coordinate[below= 0.5 cm of center] (v9);
\draw[electron] (v2) -- (v3);
\draw[photonloop] (v7) -- (v8);
\draw[photonloop] (v6) -- (v9);
\semiloopt[electronf]{v4}{v5}{0};
\semiloopt[electronf]{v5}{v4}{180};
\node[below=1.40 cm of v4] {\textit{(a)}};
\end{tikzpicture}
\hspace{2cm}   
  \begin{tikzpicture}[node distance=1cm and 1cm]
\coordinate[] (v1);
\coordinate[below=0.55 cm of v1] (v2);
\coordinate[above=0.55 cm of v1] (v3);
\coordinate[below=0.55 cm of v2] (v4);
\coordinate[above=0.55 cm of v3] (v5);
\semiloop[photonloop]{v3}{v2}{270};
\draw[electronmiddle] (v4) -- (v5);
\coordinate[right= 0.80 cm of v1] (v21);
\coordinate[below=0.55 cm of v21] (v22);
\coordinate[above=0.55 cm of v21] (v23);
\coordinate[below=0.55 cm of v22] (v24);
\coordinate[above=0.55 cm of v23] (v25);
\coordinate[right=0.50 cm of v21] (v28);
\coordinate[right=0.60 cm of v21] (v26);
\coordinate[right=0.40 cm of v26] (center2);
\coordinate[below=0.40 cm of center2] (v29);
\coordinate[below=0.70 cm of v29] (v30);
\coordinate[right= 0.80 cm of v26] (v27);
\coordinate[left= 0.20 cm of v22] (v28);
\semiloop[electronff]{v26}{v27}{0};
\semiloop[electronff]{v27}{v26}{180};
\draw[photonloop] (v26)--(v27);
\draw[photonloop] (v30) -- (v29);
\node[below=0.90 cm of v22] {\textit{(b)}};
\end{tikzpicture}
\caption{Two Feynman diagrams corresponding to one external electron, $N=1$, one external photon, $p=1$ and five vertices, $k=5$. Note that diagram (b) is disconnected.} \label{fig:npk}
\end{center}   
   \end{figure}

\noindent It is possible to evaluate explicitly the derivatives  using 
\ba 
 \nonumber
&&\frac{d^n}{dJ^n} e^{\frac{J^2}{2}} \biggr|_{J=0} = \left\{\begin{array}{cc}(n-1)!! & \text{ if n is even} \\0 & \text{ if n is odd}\end{array}\right. ,  \\ \label{Kronecker}  \\ \nonumber
&&\left(\frac{d}{d \eta }\right)^m \left( \frac{d}{ d \eta^*} \right)^n ~ e^{\eta \eta^* } \biggr|_{\eta=\eta^* = 0 }  =  n!  ~\delta_{n,m} \, ,
\ea
where the convention $(-1)!!=1$ is assumed. 
This leads to 
\ba  \nonumber Z_{N,p} \left( \lambda \right)
&=& \left\{\begin{array}{cc}\sum_{k=0}^\infty ~ \frac{(2k+N)! \, (2k+p-1)!!}{(2k)!} ~ \lambda^{2k} & \text{ if } p \text{ is even}, \\ 
\\
\sum_{k=0}^\infty ~ \frac{(2k+N+1)! \, (2k+p)!!}{(2k+1)!} ~ \lambda^{2k+1}  & \text{ if } p \text{ is odd.}\end{array}\right.
\ea

\noindent For example the coefficient of $\lambda^5$ in $Z_{1,1}$ is $90$ and the Feynman diagrams of Figure (\ref{fig:npk}) represent two of these ninety diagrams. 

\vskip 0.2cm

\noindent As we saw in Section 4,  the photon lines in the Feynman diagrams map to edges of the rooted maps. We are therefore interested only in Feynman diagrams with no external photon lines and $N$ external electron lines, in which case the number of vertices is equal to twice the number of photon lines, $k=2 \e$, and   the formula simplifies to 

\ba  {\cal{Z}}_N (\lambda) &:= &  Z_{N,0} \left( \lambda \right) \nonumber \\ 
&=&  \frac{1}{Z(0)} ~ \left(\frac{d}{d \eta} \frac{d}{ d \eta^*} \right)^N  \, Z \left(J,\eta,\eta^*, \lambda \right) \Biggr|_{J= \eta=\eta^*= 0 } \label{zndef} \\ 
\label{Zc0}
&=&  \sum_{k=0}^\infty ~ \frac{(2k+N)! \, (2k-1)!!}{(2k)!} ~ \lambda^{2k}~.~~~~ 
  \ea

\noindent As noted before, these generating functions produce both connected and disconnected diagrams. Only connected diagrams are relevant, both for physics and for the enumeration of graphs. These can be obtained using the following standard trick of taking the natural logarithm of  Eq.(\ref{part2}) before taking the derivatives with respect to the sources,  
 see for examples \cite{zvonkin} or \cite{harary}.  We will use the notation $W_{N,p}$ for  the generating functions of the connected diagrams with $N$ external electron lines and $p$ external photon lines. They   are given by   
\be \nonumber
W_{N,p} \left(\lambda \right) = \frac{1}{p! N!} \left(\frac{d}{dJ} \right)^p \left(\frac{d}{d \eta} \frac{d}{d \eta^*} \right)^N \,   \ln \left( \frac{ Z\left(J,\eta,\eta^*,\lambda\right)}{Z(0) } \right)  \biggr|_{J=\eta=\eta^*=0}
%\label{genfunc}
 \, .
\ee

\noindent As mentioned previously, we focus on the diagrams with no external photon lines, corresponding to $W_{N,0}(\lambda)$. These are  the generating functions of the connected Feynman diagrams with  $N$ external electron lines and from Theorem  \ref{thm_bijection} , they are equal to $M_N (\lambda)$ of Eq.(\ref{mnz}). We may therefore write

\ba M_N (\lambda)  &=&  W_{N,0} \left( \lambda \right) \nonumber \\
&=&  \frac{1}{N!}
 \left(\frac{d}{d \eta} \frac{d}{d \eta^*} \right)^N \, \,   \ln \left( \frac{ Z\left(J,\eta,\eta^*,\lambda\right)}{Z(0) } \right)  \biggr|_{J=\eta=\eta^*=0}
 \label{wndef} \, . \ea

\noindent In the following we will not indicate explicitly the $\lambda$ dependence of the ${\cal{Z}}_{N}$ and $M_N$ to ease the notation.

\section{The generating function of one-rooted maps}
\label{sect_onerooted_count}

\noindent The first quantity of interest for us is  the generating function of the one-rooted maps $M_1$. This is given by 
\ba
 \label{series}
M_1&=&  \frac{d^2}{d \eta^* d \eta} \,  \ln Z \left(J,\eta,\eta^*\right) \biggr|_{J=\eta=\eta^*=0} \nonumber \\
&=& \frac{{\cal{Z}}_1}{{\cal{Z}}_0} \, \nonumber \\
 &=&  1 + 2  \, \lambda^2 + 10 \,  \lambda^4 + 74 \, \lambda^6 + 706 \, \lambda^8 + 8\,162 \,  \lambda^{10}+ 110 \, 410 \, \lambda^{12} +  \ldots 
\ea
where we have used Eqs.(\ref{Zc0}) and \eqref{Kronecker}.

\subsection{The number of one-rooted maps with $\e$ edges}
\label{sect_onerooted}

\noindent Even though one can Taylor expand the ratio of two sums ${{\cal{Z}}_1}$ and ${{\cal{Z}}_0}$ as above to obtain the power series \eqref{series} and then to read off the numbers $m_1(\e)$ of one-rooted maps with  $\e$ of edges from the coefficients of the series, 
this does not give the closed form expression for these coefficients.
In order  to recover Eq.\eqref{M1e} of Arqu\`es-B\'eraud for $m_1(\e)$,  we need to express our result for $ M_1$ as a single sum over powers of $\lambda$ instead of a ratio of two sums ${{\cal{Z}}_1}$ and ${{\cal{Z}}_0}$ . To achieve this,  let us first express ${{\cal{Z}}_1}$ in terms 
 of ${\cal{Z}}_0$ and  its derivative  with respect to the coupling constant. This is done by noting that 
\ba {\cal{Z}}_1  &=& \sum_{n=0}   (2n+1) !! \lambda^{2n} \nonumber \\
&=& \left( \lambda \frac{d}{d\lambda} + 1 \right) \sum_{n=0} (2n-1)!! \lambda^{2n} \nonumber  \\ \nonumber
&=& \lambda  \,  {\cal{Z}}_0'+ {\cal{Z}}_0 \,   \label{zunzero}
\ea 
where a prime indicates a derivative with respect to the coupling constant $\lambda$. This allows us to write
\ba M_1 =  \frac{{\cal{Z}}_1}{{\cal{Z}}_0} \, 
=  \lambda \frac{{\cal{Z}}_0' }{{\cal{Z}}_0} + 1  \, . \label{zop} \ea

\noindent To recover Eq.(\ref{M1e}),  it proves convenient to rewrite this in terms of 
\be \label{mzero} M_0=\ln {\cal{Z}}_0   \, .\ee
Doing so is also useful  for the next section where we will obtain equations relating the various $M_N$ directly.
Using the previous two equations,  we write 
\ba  M_1
&=& \lambda  \,  M_0' + 1 \, . \label{w10}\ea

\noindent The connection with Eq.(\ref{M1e}) is made by using the  identity
\ba \ln  \biggl( 1+ \sum_{n=1}^\infty A_n \lambda^{2n}  \biggr)&=& \sum_{i=1}^\infty  ~\lambda^{2i} ~\sum_{k=1}^i  \frac{(-1)^{k+1}}{k} \, 
  \sum_{\substack{\mu_1 + \ldots + \mu_k=i \\ \mu_i \neq 0} } ~\prod_{j=1}^{k}  \, ~A_j~ ,
\ea giving us
\be  M_0=  \Biggl(
 \sum_{\e=1}^\infty  ~\lambda^{2\e} ~\sum_{k=1}^\e  \frac{(-1)^{k+1}}{k} \, 
  \sum_{\substack{\mu_1 + \ldots + \mu_k=e \\ \mu_i \neq 0  } } ~\prod_{j=1}^{k}  \, ~(2 \mu_j-1)!! \Biggr)  - \ln Z(0) \, .  \label{logid}
\ee

\noindent Using Eqs.(\ref{w10}) and (\ref{logid}), we get 
\ba M_1 &=& 
1 +  \sum_{\e=1}^\infty  \, (2\e) ~\lambda^{2\e} ~\sum_{k=1}^\e  \frac{ (-1)^{k+1}}{k} \, 
  \sum_{\substack{\mu_1 + \ldots + \mu_k=\e \\ \mu_i \neq 0 } } ~\prod_{j=1}^{k}  \, ~(2 \mu_j-1)!!
\nonumber \\ \nonumber
&=& 
 \sum_{\e=0}^\infty  ~\lambda^{2\e} ~\sum_{k=0}^\e   (-1)^{k} \, 
  \sum_{\substack{\mu_1 + \ldots + \mu_{k+1}=\e+1 \\ \mu_i \neq 0} } ~\prod_{j=1}^{k+1}  \, ~(2 \mu_j-1)!! \label{w10fin} \, .
\ea

\noindent From this we can read off the number of rooted maps with $\e$ edges  since by  Eq.\eqref{mnz} we have
$ \nonumber M_1 = \sum_{\e=0}^\infty m_1(\e) \lambda^{2\e} \, . $
This gives us
 \be m_1(\e) =  \sum_{k=0}^\e (-1)^k 
 \sum_{\substack{\mu_1 + \ldots + \mu_{k+1}=\e+1 \\ \mu_i \neq 0 } }  ~\prod_{j=1}^{k+1} (2 \mu_j-1)!!   \, .\label{tata}
 \ee
 We have therefore recovered the known expression for $m_1(\e)$ of Arqu\`es and B\'eraud, Eq.(\ref{M1e}),  using quantum field theory.

\subsection{Differential equation for $M_1(\lambda)$}

As mentioned in Section \ref{sect_M1}, a differential equation for the generating function of one-rooted maps $M_1$  is known. In this section we will recover it from quantum field theory.

\vskip 0.2cm

\noindent For the theory we are considering, one can derive a differential equation for $M_0$ \cite{prd}. The details are presented in the appendix and the result is 
\be
M_0' = 2 \lambda + 4 \lambda^2 M_0' + \lambda^3 M_0'' + \lambda^3 (M_0')^2 \, , \label{wprime}\ee
where a prime indicates a derivative with respect to $\lambda$.

\vskip 0.2cm

\noindent We can now use this to generate a differential equation for $ M_1 $.
Isolating $M_0'$ in terms of $M_1$  in Eq.(\ref{w10}) and plugging into Eq.(\ref{wprime}), we get 
\be
M_1 = 1  + \lambda^2 M_1 + \lambda^3  M_1'  + \lambda^2 M_1^2 , \label{diffeqm1}
\ee
which is the known equation,  Eq.(\ref{M1DifferentialEquation}).

\section{Generating functions of $N$-rooted maps}
\label{sect_Nrooted_count}

Recall  that $M_N(\lambda)$ is the generating function of $N$-rooted maps defined in \eqref{mnz} by 
\be 
M_N (\lambda)  = \sum_{\e=0}^\infty m_N(\e)  \lambda^{2\e}\nonumber  \, ,\ee
with $m_N(\e)$ being the number of $N$-rooted maps with $\e$ edges regardless of genus. 
 Theorem \ref{thm_thm2} in this section gives a closed form expression for $M_N(\lambda)$ for all values of $N$.

\subsection{A closed form expression for the generating function of $N$-rooted maps}

\begin{theorem} 
\label{thm_thm2}
The generating function $M_N(\lambda)$ of N-rooted maps, where $N\geq 1$,  is given 
 by
  \ba  M_N (\lambda) =  \sum_{\substack{\alpha_1 + 2 \alpha_2 + \ldots +N \alpha_N = N \\ \alpha_1 \ldots \alpha_N \geq 0 } }~ ~
  \frac{N!}{\alpha_1! \alpha_2 ! \ldots \alpha_N!} ~ 
    \frac{ (-1)^{\alpha_1 + \ldots +\alpha_N -1} ~ \left( \alpha_1 + \ldots  +\alpha_{N}-1 \right)!}{{\cal{Z}}_0^{\alpha_1 + \ldots + \alpha_N}} 
   ~  \prod_{1 \leq j \leq N}  \left( \frac{{\cal{Z}}_{j}}{  (j!)^2} \right)^{\alpha_j}\, , \label{mnn} \ea
  where
  \be {\cal{Z}}_j =  \sum_{k=0}^\infty ~ \frac{(2k+j)! \, (2k-1)!!}{(2k)!} ~ \lambda^{2k} \,, \qquad j\geq 0. \label{theozj} \ee
  The number $m_N(\e)$ of N-rooted maps with $\e$ edges can then be obtained using 
  \be m_N(\e) = \frac{1}{(2\e)!} \lim_{\lambda \rightarrow 0} \, \frac{d^{2\e}}{d \lambda^{2\e}} M_N (\lambda) \, . \label{derivat}\ee
 \end{theorem}
 
 \begin{proof} 
 
 We first note that $Z(J,\eta,\eta^*,\lambda)$ as given in Eq.(\ref{zees})  depends on $\eta$ and $\eta^*$ only  through their  product $\eta \eta^*$. Therefore in this proof we will write $Z(J,\eta\eta^*,\lambda)$ instead of $Z(J,\eta,\eta^*,\lambda)$. 
 For a differentiable  function depending only on the product of two variables $\eta$ and $\eta^*$, we have
 \be \left( \frac{d^2}{d\eta d\eta^*} \right)^N \, \ln   f(\eta \eta^*) \biggr|_{\eta=\eta^*=0} = N!  \, \frac{d^N}{dx^N}  \ln f(x) \biggr|_{x=0} \, ,\label{chder} \ee
 which is valid at the condition that $f(0) \neq 0$. 
 We can apply this identity to Eq.(\ref{wndef}) with $f(\eta\eta^*) = Z(J,\eta \eta^*,\lambda)$, which does not vanish at $\eta \eta^*=0$, to obtain
 
 \be M_N (\lambda)  =   
 \, \frac{d^N}{d x^N}  \,   \ln  \Bigl( \frac{Z\left(J,x,\lambda \right)}{Z(0)} \Bigr) \biggr|_{J=x=0}  \, ,\nonumber \,  \ee

where it is understood that the product $\eta \eta^*$ has been replaced by $x$.  

 \vskip 0.2cm
 
\noindent For $ N \geq 1$, 
  we apply Fa\`a di Bruno's formula  to the Nth derivative of  the logarithm of a function, giving us 
  
\ba  M_N (\lambda) & =&  \sum_{\substack{\alpha_1 + 2 \alpha_2 + \ldots +N \alpha_N = N \\ \alpha_1 \ldots \alpha_N \geq 0 } }~
  \frac{N!}{\alpha_1! \alpha_2 ! \ldots \alpha_N!} ~  
  \frac{ (-1)^{\alpha_1 + \ldots +\alpha_N -1} ~ \left( \alpha_1 + \ldots  +\alpha_{N}-1 \right)!}{ 
   {\cal{Z}}_0(\lambda)^{\alpha_1 + \ldots + \alpha_N}} 
    \times 
  \nonumber \\ &~& \qquad   \qquad   \qquad \qquad \qquad \qquad \qquad \qquad \qquad  \qquad  \qquad 
  ~  \prod_{j=1}^N  \left( \frac{1}{Z(0)} \frac{1}{ j!} \frac{d^j}{dx^j}  Z(J,x,\lambda)  \right)^{\alpha_j}  \Biggr|_{J=x=0}\, , \nonumber
   \ea
   where we have used the notation of Eq.(\ref{zndef}) to write
   \be \nonumber  \frac{Z(J,x,\lambda)}{Z(0)}  \Bigr|_{J=x=0} = {\cal{Z}}_0(\lambda) \,  . \ee
 Now we use once more Eq.(\ref{chder}) to rewrite the expression   in terms of $\eta$ and $\eta^*$ instead of x, giving us 
 \ba &~&  M_N (\lambda)  =   \sum_{\substack{\alpha_1 + 2 \alpha_2 + \ldots +N \alpha_N = N \\ \alpha_1 \ldots \alpha_N \geq 0 } }~
  \frac{N!}{\alpha_1! \alpha_2 ! \ldots \alpha_N!} ~
   \frac{ (-1)^{\alpha_1 + \ldots +\alpha_N -1} ~ \left( \alpha_1 + \ldots  +\alpha_{N}-1 \right)!}{   
   \, {\cal{Z}}_0(\lambda)^{\alpha_1 + \ldots + \alpha_N}} 
  \nonumber \\ &~&  \qquad   \qquad  \qquad  \qquad \qquad \qquad \qquad  \qquad  \qquad    \times  \prod_{j=1}^{  N}  \left( \frac{1}{Z(0)}  \frac{1}{  (j!)^2} \frac{d^{2j}}{d \eta^j d \eta^{*j}}
 Z(J,\eta,\eta^*,\lambda)   \right)^{\alpha_j}  \Biggr|_{J=\eta=\eta^*=0}\, .  \nonumber
   \ea
 Using again the notation of Eq.(\ref{zndef}), we obtain Eq.(\ref{mnn}) (where the dependence on $\lambda$ of   ${\cal{Z}}_0$ and ${\cal{Z}}_j$  is omitted). As for the expression for ${\cal{Z}}_j$ of Eq.(\ref{theozj}), it is already calculated in Eq.(\ref{Zc0}).

\end{proof}

    \subsection{An algorithm to derive a formula for $m_N(\e)$}
  
Here we discuss an alternative to Eq.\eqref{derivat} from Theorem \ref{thm_thm2} way to compute generating functions for $N$-rooted graphs. This is a generalization of the approach used in Section \ref{sect_onerooted} to obtain $M_1(\lambda)$ and $m_1(e)$. Recall that the first step in Section \ref{sect_onerooted} was to express ${{\cal{Z}}_1}$ in terms  of ${\cal{Z}}_0$ and  its derivative  with respect to $\lambda$. 
It is also simple to write an explicit expression for ${\cal{Z}}_N$ in terms of ${\cal{Z}}_0$ and derivatives of ${\cal{Z}}_0$. From expression (\ref{Zc0}) for  $\mathcal Z_N$, we have
 \ba {\cal{Z}}_N  &=&\sum_{k=0}^\infty \frac{(2k+N)! (2k-1)!!}{(2k)!} \lambda^{2k} \,  \nonumber \\
 &=& \left(N + \lambda \frac{d}{d\lambda} \right) Z_{N-1}  \, ,\nonumber \ea
 so that 
 
 \ba  {\cal{Z}}_{N} &=& \left(N+ \lambda \frac{d}{d \lambda} \right) \left(N-1+ \lambda \frac{d}{d \lambda} \right) \ldots 
\left(1+ \lambda \frac{d}{d \lambda} \right) {\cal{Z}}_0  \,  \nonumber \\
&=&\sum_{k=0}^N  \binom{N}{k} \frac{N!}{k!} ~ \lambda^k ~ \frac{d^k {\cal{Z}}_0 }{d\lambda^k}  \,   \label{hgt} \\
&=&  \frac{d^N}{d \lambda^N}  \bigl( \lambda^N Z_0 \bigr)  \, .    \label{general}   \ea

\noindent Using this, in the remaining part of this section, we compute $M_2$ and $M_3$.

\begin{paragraph}{Calculating $M_2$.}
  
  From Eq.(\ref{mnn}) for $M_N(\lambda)$ from Theorem \ref{thm_thm2}, we have 

 \be M_2 =\frac{1}{2} \frac{{\cal{Z}}_2}{{\cal{Z}}_0} - \biggl( \frac{{\cal{Z}}_1}{{\cal{Z}}_0}\biggr)^2 \,.  \label{shs}\ee
\noindent Using again expression \eqref{Zc0} for $\mathcal Z_N$, this leads to the following Taylor expansion

 \ba \nonumber M_2 &=&
     \lambda^2 + 13 \,   \lambda^4 + 165 \,   \lambda^6 + 2\,273 \,   \lambda^8 + 34 \,577  \,  \lambda^{10} +
 581 \, 133 \, \lambda^{12} \ldots\;.
  \ea

\noindent In order to obtain  an explicit expression for $ m_2(e)$  
we may rewrite formula \eqref{shs} for $M_2 $ 
as an expansion in powers of $\lambda$ symbolically. From \eqref{general} we have

\begin{equation*}
{\cal{Z}}_2 =  \lambda^2 {\cal{Z}}_0''+ 4 \lambda {\cal{Z}}_0'+  2{\cal{Z}}_0 
\qquad\mbox{and}\qquad
{\cal{Z}}_1  =  \lambda{\cal{Z}}_0'+ {\cal{Z}}_0 \,.
\end{equation*}

\noindent With this and Eq.\eqref{mzero}, Eq.(\ref{shs}) becomes
   \ba M_2 
   &=& \frac{1}{2}   \frac{\lambda^2}{{\cal{Z}}_0}  {\cal{Z}}_0''-  \frac{\lambda^2}{{\cal{Z}}_0^2} \bigl(   {\cal{Z}}_0'\bigr)^2 \, ,
   \nonumber \\
   &=&
   \frac{\lambda^2}{2}  M_0''- \frac{\lambda^2}{2} \left( M_0' \right)^2  \, .
   \label{w20log}
   \ea
  
  \noindent  Using  the explicit expression (\ref{logid}) for $M_0$ and the relation $\lambda  \,  M_0' =M_1- 1$  from (\ref{w10}),  we obtain 
  \ba M_2 &=&
   \sum_{e=1}^\infty  \Biggl[
    \,  e (2e-1) ~\lambda^{2e} ~\sum_{k=1}^e  \frac{(-1)^{k+1}}{k} \, 
  \sum_{\substack{\mu_1 + \ldots + \mu_k=e \\ \mu_i \neq 0 } } ~\prod_{j=1}^{k}  \, ~(2 \mu_j-1)!!  
  \Biggr] - \frac{1}{2} \bigl(M_1-1 \bigr)^2  \, .
  \ea

\noindent Rewriting this in terms of $  M_1
= \sum_{k=0} m_1(k) \lambda^{2k}$ with $m_1(k)$ given in Eq.(\ref{tata}),
 we obtain

\ba
M_2 &=& 
   \lambda^2 +
  \sum_{e=2}^\infty \lambda^{2e} \Biggl[
   \,  e (2e-1)  ~\sum_{k=1}^e  \frac{(-1)^{k+1}}{k} \, 
  \sum_{\substack{\mu_1 + \ldots + \mu_k=e \\ \mu_i \neq 0 } } ~\prod_{j=1}^{k}  \, ~(2 \mu_j-1)!!  
- \frac{1}{2} \sum_{k=1}^{e-1} m_1(k) m_1(e-k)  \Biggr]  \nonumber.
  \ea

\noindent This may be written as 
\ba
M_2 &=& 
\lambda^2  + 
\sum_{e=2}^\infty \lambda^{2e} \Biggl[  \sum_{k=0}^e (-1)^k 
 \sum_{\substack{\mu_1 + \ldots + \mu_{k+1}=e+1 \\ \mu_i \neq 0 } }   \mu_{k+1} ~\prod_{j=1}^{k+1}  \,~(2 \mu_j-1)!!  
-   ~\frac{1}{2} 
  \sum_{k=1}^{e-1} m_1(k) m_1(e-k)  \Biggr]  \, .  \qquad\label{w20}
\ea

\noindent Now we can read off the values of $m_2(e)$ from Eq.(\ref{w20}) since by definition 
(\ref{mnz}) we have
$ M_2 = \sum_{e=0}^\infty m_2(e) \, \lambda^{2e}$.

 \end{paragraph}

 \begin{paragraph}{Calculating $M_3 $.}

As an another example, one finds from Eq.(\ref{mnn}) of Theorem \ref{thm_thm2}
\ba  M_3   &=& 
\frac{1}{6}  \frac{{\cal{Z}}_3}{{\cal{Z}}_0}-   \frac{3}{2} \frac{{\cal{Z}}_1 {\cal{Z}}_2}{{\cal{Z}}_0^2}  + 2 \Biggl(\frac{{\cal{Z}}_1}{{\cal{Z}}_0} \Biggr)^3  
 \nonumber \\
 &=&  6 \,  \lambda^4 + 172 \,  \lambda^6 + 3\,834 \,  \lambda^8  + 81 \, 720  \, \lambda^{10} + 
 1\, 775 \, 198 \, \lambda^{12} + 
  \ldots   \;.
 \label{w30}
 \ea
  
\noindent We may express ${\cal{Z}}_3$ in terms of ${\cal{Z}}_0$ using Eq.(\ref{general}):
 
\ba {\cal{Z}}_3 
 &=& \lambda^3 {\cal{Z}}_0''' + 9 \lambda^2 
  {\cal{Z}}_0'' + 18 \lambda   {\cal{Z}}_0' + 6  {\cal{Z}}_0,
 \ea
giving us 
\ba M_3 &=&  2 \lambda^3 \left( \frac{{\cal{Z}}_0'}{{\cal{Z}}_0}  \right)^3 -  \frac{3}{2} \lambda^3  \frac{{\cal{Z}}_0' {\cal{Z}}_0''}{{\cal{Z}}_0^2} + \frac{\lambda^3}{6} \frac{ {\cal{Z}}_0'''}{{\cal{Z}}_0} \, , \nonumber
\\ &=&
\frac{2}{3} \lambda^3 \bigl( M_0' \bigr)^3 - \lambda^3 M_0' \, M_0''+ \frac{\lambda^3}{6}  M_0'''. \label{bigw30}
\ea

\noindent This could be  written as an explicit expansion in $\lambda$  using again Eq.(\ref{logid}).
\end{paragraph}

\section{Relating  generating functions of $N$-rooted maps to $M_1$}
\label{sect_equation}

Recall that the generating function $M_1(\lambda)$ of one-rooted maps satisfies differential equation \eqref{diffeqm1}. As a generalization of this equation to the case of $N$-rooted maps, we find that all generating functions $M_N(\lambda)$ can be expressed in terms of $M_1(\lambda)$.

\begin{theorem} 
\label{thm_thm3}
The generating function for N-rooted maps $M_N(\lambda)$  can be expressed as a polynomial of degree $N$ in the  generating function $M_1(\lambda)$ for one-rooted maps. This polynomial has $\lambda$-dependent coefficients and is obtained by substituting the following expression for $\mathcal Z_j/\mathcal Z_0$ into Eq.(\ref{mnn}) of Theorem \ref{thm_thm2}:
\ba \frac{\mathcal Z_j}{\mathcal Z_0} &=& \sum_{n=0}^j \binom{j}{n} \frac{j!}{n!} \sum_{k=0}^n (-1)^{n-k} ~ B_{n,2k-1}  \nonumber \\ &~&~~~~~~\times\Biggl[ \delta_{k,0} + H(k-1) ~\frac{M_1}{\lambda^{2k-2}} - H(k-2) \, (1- M_1 \lambda^2)
\,   \sum_{m=0}^{k-2} \frac{(2m+1)!!}{\lambda^{2k-2m-2}}   \Biggr]\,.\qquad
\label{theorem3} \ea 
Here $H$ is the Heaviside function with the convention $H(k) = 1 $ for $ k \geq 0$ and the coefficients $B$ are obtained from the recursion formula 
\ba
B_{n+1,2k-1} = B_{n,2k-3} + (2k+n+1) B_{n,2k-1},~~~~n \geq 1,~~~~n \geq k \geq 0  \label{recur}
\ea
with initial conditions 
\ba &~& B_{0,-1} =B_{1,-1} = B_{1,1} = 1, \label{ini1} \\
&~& B_{n,-3} = B_{n,2n+1} = 0. \label{ini2} \ea

\end{theorem}

\begin{remark}
 The recursion formula (\ref{recur})  can be solved for each given value of $k$. For example,
\ba B_{n,-1} &=& n!, \nonumber \\
B_{n,2n-1} &=& 1, \nonumber \\
B_{n,2n-3} &=& \frac{(3n-1)n}{2}. \nonumber \ea
\end{remark}

\begin{proof} We only  need to prove relation \eqref{theorem3} for $\mathcal Z_j/\mathcal Z_0$. Throughout this proof we will indicate explicitly the $\lambda$ dependence of the various  quantities.
We first define the following quantity for odd $i$ only
\be R_i(\lambda) := \sum_{k=0}^\infty (2k + i)!! ~ \lambda^{2k} ~,~~~i \geq -1,~~\text{ odd } i .\ee
In particular, 
\be R_{-1}(\lambda) = \mathcal Z_0 (\lambda) \, . \label{rone}  \ee 
We may express all the other $R_i(\lambda)$ in terms of $\mathcal Z_0$. It is easy to check that
\be R_1(\lambda) = \frac{ \mathcal Z_0(\lambda) -1}{\lambda^2}  \, ,  \label{rtwo} \ee and 
\be R_j(\lambda) = ~\frac{\mathcal Z_0(\lambda)-1}{\lambda^{j+1}} - \sum_{m=0}^{\frac{j-3}{2}} \frac{(2m+1)!!}{\lambda^{j-2m-1}}, ~~~~~ j \geq 3 .   \label{rthree} \ee
It follows that
\be \lambda \frac{d R_k(\lambda)}{d \lambda} =  R_{k+2}(\lambda) - (k+2) R_k(\lambda),  ~~~~~~k \geq -1\, . \label{deri} \ee

\noindent We now consider $\lambda^n \mathcal Z_0^{(n)}(\lambda)$ where the  index $(n)$ indicates the number of derivatives with respect to $\lambda$. For example 
\ba \lambda \mathcal Z_0^{(1)} &=& \lambda \frac{d}{d\lambda} R_{-1} (\lambda) \nonumber \\
&=& 	R_1 (\lambda) - R_{-1} (\lambda) \label{roog}\,. \ea  
We see that $\lambda^n \mathcal Z_0^{(n)}$ will contain $n+1$ terms and will take the form
\be
\lambda^n \mathcal Z_0^{(n)}  (\lambda) = \sum_{k=0}^n  (-1)^{n-k} \, B_{n,2k-1} ~R_{2k-1} (\lambda)\,,  \label{oop} \ee
where the  sign has been introduced  to make the coefficients $B_{n,2k-1}$ positive.  We obtain a recursion formula for the coefficients $B$ by taking a derivative of both sides with respect to $\lambda$ and then multiplying the result by one power of $\lambda$. This gives
\be n \lambda^n  \mathcal Z_0^{(n)}  (\lambda) + \lambda^{n+1} \mathcal Z_0^{(n+1)}  (\lambda) =\sum_{k=0}^n  (-1)^{n-k} \,  B_{n,2k-1} ~ \lambda \frac{d}{d \lambda} R_{2k-1} (\lambda) \;.
\ee
 Using Eq.(\ref{oop})  for the terms on the left side and Eq.(\ref{deri}) to evaluate the derivative of $R_{2k-1} (\lambda)$, we obtain the  recursion formula \eqref{recur}
 with initial conditions \eqref{ini1}, \eqref{ini2}.
 
 \vskip 0.2cm
 
\noindent  Using Eqs.(\ref{rone}), (\ref{rtwo}), and (\ref{rthree}), we write Eq.(\ref{oop}) as
 \be
 \lambda^n \mathcal Z_0^{(n)}  (\lambda) = \sum_{k=0}^n  (-1)^{n-k} \, ~B_{n,2k-1} (\lambda) \Bigl[\delta_{k,0} ~\mathcal Z_0(\lambda) + H(k-1) ~\frac{\mathcal Z_0 (\lambda) -1}{\lambda^{2k}}  -H(k-2) \sum_{m=0}^{k-2} \frac{(2m+1)!!}{\lambda^{2k-2m-2}} \Bigr]\;. \label{fff}
 \ee
  From Eq.(\ref{zop}), we have
 \be M_1 =  \lambda \frac{\mathcal Z_0' (\lambda)}{\mathcal Z_0 (\lambda)} +1 \, . \label{m1der}\ee
The derivative $\mathcal Z_0'(\lambda)$ can be expressed in terms of $\mathcal Z_0 (\lambda)$ using Eqs.(\ref{rone}), (\ref{rtwo}), and (\ref{roog}). This leads to 
 \be\nonumber
 \mathcal Z_0' (\lambda) = \frac{ \mathcal Z_0 (\lambda) - 1 - \lambda^2 \mathcal Z_0 (\lambda) }{\lambda^3} \, .
 \ee
 Using this result and Eq.(\ref{m1der}), 
 we may write 
 \ba \mathcal Z_0 (\lambda)  - 1  =  \lambda^2 M_1(\lambda) \mathcal Z_0 (\lambda)  \label{abv1}\ea
 and therefore \ba \frac{1}{\mathcal Z_0(\lambda) } = 1 - \lambda^2  M_1 (\lambda)  \label{abv2} \, .\ea
Using Eqs.(\ref{abv1}), (\ref{abv2}), and (\ref{fff})  we obtain

\ba
\frac{1}{\mathcal Z_0(\lambda) } \, \lambda^n \mathcal Z_0^{(n)}  (\lambda)&=& \sum_{k=0}^n  (-1)^{n-k} \, ~B_{n,2k-1} 
 \nonumber \\
&~&~~~\times\Bigl[\delta_{k,0} ~ + H(k-1) ~\frac{M_1 (\lambda) }{\lambda^{2k-2}}  -H(k-2)\,\bigl(1- \lambda^2 M_1(\lambda) \bigr)  \sum_{m=0}^{k-2} \frac{(2m+1)!!}{\lambda^{2k-2m-2}} \Bigr] \;.\qquad \label{just}
 \ea

\noindent Recall that from expression (\ref{hgt}) of $\mathcal Z_N$ in terms of derivatives of $\mathcal Z_0$, we have
 \be  \frac{\mathcal Z_j (\lambda) }{\mathcal Z_0(\lambda)} = \sum_{n=0}^j \binom{j}{n} \frac{j!}{n!} \,   \frac{1}{\mathcal Z_0(\lambda)} \, \lambda^n \, 
  \mathcal Z_0^{(n)}(\lambda) \, . \nonumber 
 \ee
 Now inserting Eq.(\ref{just}) into this last equation, we obtain the required expression (\ref{theorem3}) for the ratio $\mathcal Z_j/\mathcal Z_0$.
 
 \noindent It is now a simple matter to show that the result for $M_N$ is a polynomial of degree $N$ in $M_1$.  From Eq.(\ref{theorem3}) we see that $  \mathcal Z_j / \mathcal Z_0$ is a polynomial of order one in $M_1$, for any value of $j$.   According to Eq.(\ref{mnn}), $M_N(\lambda)$  thus contains a linear combination of terms of the form 
 \be M_1^{\alpha_1 + \alpha_2 + \ldots + \alpha_N},  \nonumber \ee
 with the condition $\alpha_1 +2  \alpha_2 + \ldots + N \alpha_N = N$. The largest power of $M_1$ is thus found by maximizing the sum  
 $\alpha_1 + \ldots + \alpha_ N$ while satisfying this condition. It is clear that this corresponds to choosing $\alpha_1=N$ and all the other indices equal to zero, giving us that $M_N$ is a polynomial of order $N$.
 \end{proof}
 
 \begin{example}

Here are the generating functions of $N$-rooted maps in terms of $M_1$ for $N=1,\dots, 5$. 
 \ba 2 \lambda^2 M_2 &=& M_1 - 1 - 2 \lambda^2 M_1^2 \, , \nonumber \\ 
 \nonumber\\
 6 \lambda^4 M_3 &=&  M_1 - 1 - 9 \lambda^2 M_1^2 + 7 \lambda^2 M_1 + 12 \lambda^4  M_1^3 \, , \nonumber \\
  \nonumber\\
 24 \lambda^6 M_4 &=&  M_1-1 - 15 \lambda^2 + 47 \lambda^2 M_1 -34 \lambda^2 M_1^2 -112 \lambda^4 M_1^2 +144 \lambda^4 M_1^3 -144 \lambda^6  
 M_1^4 \, , \nonumber \\
  \nonumber\\
 120 \lambda^8 M_5 &=& M_1-1 - 93 \lambda^2 + 216 \lambda^2 M_1 + 633 \lambda^4 M_1 -125 \lambda^2 M_1^2  \nonumber \\ &~&~~ - 1  \,875 \lambda^4 M_1^2 
+ 1  \,300 \lambda^4 M_1^3   + 2  \,800 \lambda^6 M_1^3 - 3  \,600 \lambda^6 M_1^4 + 2  \, 880 \lambda^8 M_1^5 \, . \nonumber 
\ea
 
 \end{example}

\appendix\markboth{}{}
\renewcommand{\thesection}{\Alph{section}}
\numberwithin{equation}{section}
\section{Appendix}
\section*{The differential equation for $M_0$}
\label{sect_appendix}
Here we present the derivation of the  differential equation for $M_0$ given in  Eq.(\ref{wprime}) using quantum field theory. Let us recall that $M_0$ is defined by (see Eq.(\ref{wndef}))
\be
M_0 (\lambda)  =   
 \,   \ln \left( \frac{ Z\left(J,\eta,\eta^*,\lambda\right)}{Z(0) } \right)  \biggr|_{J=\eta=\eta^*=0} \, , \label{mno}
\ee
where $Z(0) = \pi \sqrt{2 \pi} $. 

\noindent Our starting point is Eq.(\ref{partfunc}) which may be written as

\ba Z (J,\eta,\eta^*,\lambda) &=&  \pi  \sqrt{ 2 \pi} \, \sum_{k=0}^\infty \frac{\lambda^k}{k!}   \left( \hat{A} \hat{\phi} \hat{\phi}^* \right)^k ~ \exp \left( \eta \eta^* + \frac{J^2}{2} \right)  \nonumber  \\
&=&  \pi  \sqrt{ 2 \pi} \, \exp \left(\lambda  \hat{A} \hat{\phi} \hat{\phi}^* \right)~ \exp \left( \eta \eta^* + \frac{J^2}{2} \right), \label{eee}
\ea where we have defined the operators
\ba \hat{A} &:=& \frac{d}{dJ}, \nonumber  \\
\hat{\phi} &:=& \frac{d}{d\eta}, \nonumber  \\
\hat{\phi}^* &:=& \frac{d}{d \eta^*} \, . \nonumber 
\ea
\noindent Our goal is to obtain a differential equation for the partition function $Z$ and its derivatives with respect to the coupling constant $\lambda$. Let us then consider the derivative of $Z$ with respect to $\lambda$:
\be  \frac{d Z (J,\eta,\eta^*,\lambda)}{d \lambda} =  \hat{A} \hat{\phi} \hat{\phi}^* ~ Z(J,\eta,\eta^*,\lambda)  \, . \label{zprime}  \ee
\noindent The next step is to express the right hand side in terms of $Z$ and its derivatives. For this, we begin by calculating the application of $\hat{\phi}$ on the partition function:
\ba \hat{\phi} ~Z(J,\eta,\eta^*,\lambda)   &=&
 \pi  \sqrt{ 2 \pi} \, \exp \left(\lambda  \hat{A} \hat{\phi} \hat{\phi}^* \right)~\hat{\phi}~ \exp \left( \eta \eta^* + \frac{J^2}{2} \right)  \nonumber  \\ 
&=&  \pi  \sqrt{ 2 \pi}
\sum_{k=0}^\infty  \frac{    \left(  \lambda \hat{A} \hat{\phi} \hat{\phi}^* \right)^k   }{k!} 
 \,~\eta^*~ \exp \left( \eta \eta^* + \frac{J^2}{2} \right)  \nonumber  \\ 
 &=& 
  \pi  \sqrt{ 2 \pi}
\sum^\infty _{k=0} \frac{   \left( \lambda  \hat{A} \hat{\phi} \right)^k  }{k!}   ~
\left( k ( \hat{\phi}^*)^{k-1} + \eta^* (\hat{\phi}^*)^k \right) ~ \exp \left( \eta \eta^* + \frac{J^2}{2} \right) \nonumber   \\ 
&=&
 \pi  \sqrt{ 2 \pi} \bigl(  \lambda \hat{A} \hat{\phi}  + \eta^* \bigr) 
\sum^\infty _{k=0} \frac{\lambda^k}{k!}   \left( \hat{A} \hat{\phi}  \hat{\phi}^* \right)^k 
~ \exp \left( \eta \eta^* + \frac{J^2}{2} \right)  \nonumber  \\ 
&=&  \bigl(  \lambda \hat{A} \hat{\phi}  + \eta^* \bigr) ~ Z(J,\eta,\eta^*,\lambda) \, . \label{ii1}
 \ea
 \noindent Following a similar approach, we find
 \ba
  \hat{\phi}^* ~Z(J,\eta,\eta^*,\lambda)   &=&   \bigl( \lambda  \hat{A} \hat{\phi}^*  + \eta \bigr) ~ Z(J,\eta,\eta^*,\lambda) \, , \label{ii2} 
  \\
   \hat{A} ~Z(J,\eta,\eta^*,\lambda)   &=&   \bigl( \lambda  \hat{\phi}  \hat{\phi}^* + J \bigr) ~ Z(J,\eta,\eta^*,\lambda) \, . \label{ii3}
   \ea
\noindent   Equations (\ref{ii1}),(\ref{ii2}) and (\ref{ii3}) are examples of Dyson-Schwinger equations.  Now consider
   \ba
    \hat{\phi}  \hat{\phi}^* ~Z(J,\eta,\eta^*,\lambda)   &=&  \hat{\phi}  \left( \lambda \hat{A} \hat{\phi}^*  + \eta 
    \right) ~ Z(J,\eta,\eta^*,\lambda) \, \nonumber   \\ &=&
     \left( \lambda \hat{A}  \hat{\phi} \hat{\phi}^*  + 1 +  \eta \hat{\phi} 
    \right) ~ Z(J,\eta,\eta^*,\lambda) \, . \label{oo1}
    \ea
 \noindent    At first sight this result may appear suspicious. Indeed, the operators $\hat{\phi}$ and $\hat{\phi}^*$ commute but the right hand side would have contained a term $\eta^* \hat{\phi}^* Z$ if we had calculated  $  \hat{\phi}^*  \hat{\phi} ~Z $ instead.   The two expressions would be different if $Z$ was an arbitrary function of $\eta$ and $\eta^*$ but they coincide when applied to the partition function of  Eq.(\ref{eee}). Of course, equations (\ref{ii1}), (\ref{ii2}) and (\ref{ii3}) are also valid only because of the form (\ref{eee}) of the partition function.
 
 \vskip 0.2cm

\noindent  We now apply $\hat{A}$ to Eq.(\ref{oo1}):
 \ba 
  \hat{A}   \hat{\phi} \hat{\phi}^*Z(J,\eta,\eta^*,\lambda) &=&
\hat{A}   \left( \lambda \hat{A}  \hat{\phi} \hat{\phi}^*  + 1 +  \eta \hat{\phi} 
    \right) ~ Z(J,\eta,\eta^*,\lambda)  \nonumber  \\
    &=&   \left( \lambda \hat{A} \hat{\phi} \hat{\phi}^*  + 1 +  \eta \hat{\phi} 
    \right) ~ \left(\lambda \hat{\phi} \hat{\phi}^* + J \right) Z(J,\eta,\eta^*,\lambda) \nonumber   \\
    &=&
     \left( \lambda \hat{A} \hat{\phi} \hat{\phi}^*  + 1 +  \eta \hat{\phi} 
    \right) ~ \left(\lambda \hat{\phi} \left(    \lambda  \hat{A} \hat{\phi}^*  + \eta  \right) + J \right) Z(J,\eta,\eta^*,\lambda) \nonumber  \\
    &=&  \left( \lambda \hat{A} \hat{\phi} \hat{\phi}^*  + 1 +  \eta \hat{\phi} 
    \right) ~ \left(  \lambda^2  \hat{A} \hat{\phi}   \hat{\phi}^*  + \lambda + \lambda  \eta \hat{\phi}  + J \right) Z(J,\eta,\eta^*,\lambda) \, , \label{yank}
    \ea
    where in the third step we have used Eq.(\ref{ii2}).
    
\noindent The only nontrivial terms to calculate are 
\ba  \lambda \hat{A} \hat{\phi} \hat{\phi}^* \left( \lambda  \eta \hat{\phi}  \right) Z(J,\eta,\eta^*,\lambda) &=& \left( \lambda^2  \hat{A} \hat{\phi} \hat{\phi}^*+
  \lambda^2 \eta \hat{A} \hat{\phi}^2 \hat{\phi}^* \right)Z(J,\eta,\eta^*,\lambda) \, , \nonumber   \\
  \lambda \hat{A} \hat{\phi} \hat{\phi}^* \, J  \, Z(J,\eta,\eta^* \, \lambda) &=&
  \left(  \lambda  \hat{\phi} \hat{\phi}^* +  \lambda J  \hat{A} \hat{\phi} \hat{\phi}^* \right) Z(J,\eta,\eta^*,\lambda)  \nonumber   \\
  &=&   \left( \lambda^2 \hat{A}  \hat{\phi} \hat{\phi}^*  + \lambda +  \lambda \eta \hat{\phi} 
  +  \lambda J  \hat{A} \hat{\phi} \hat{\phi}^*    \right) Z(J,\eta,\eta^*,\lambda)  \, ,  \nonumber \\
   \eta \hat{\phi}  \left( \lambda  \eta \hat{\phi}  \right) Z(J,\eta,\eta^*,\lambda) &=& \left(  \lambda  \eta \hat{\phi}  +  \lambda  \eta^2 \hat{\phi} 
   \right) Z(J,\eta,\eta^*,\lambda) \, . \nonumber
\ea 

\noindent  Using these relations, Eq.(\ref{yank}) may be written as
\begin{multline*}
  \hat{A}   \hat{\phi} \hat{\phi}^*Z(J,\eta,\eta^*,\lambda) =\Biggl( \lambda^3 (\hat{A} \hat{\phi} \hat{\phi}^*)^2 + 4 \lambda^2 \hat{A} \hat{\phi} \hat{\phi}^*+ 2 \lambda 
  + 2 \lambda^2  \eta \hat{A} \hat{\phi}^2 \hat{\phi}^* 
  \\
  + 4 \lambda  \eta \hat{\phi} + \lambda \eta^2  \hat{\phi} + J \eta  \hat{\phi}  + \lambda J \hat{A} \hat{\phi} \hat{\phi}^* + J \Biggr) Z(J,\eta,\eta^*,\lambda) \, . \nonumber 
\end{multline*}
\noindent Setting now all the sources to zero and using Eq.(\ref{zprime}), one finds
\be
\frac{d Z(\lambda)}{d\lambda}  =\lambda^3  \frac{d^2 Z}{d\lambda^2} (\lambda)  + 4 \lambda^2 \frac{d Z(\lambda) }{d \lambda}  + 2 \lambda Z(\lambda) \, ,
 \label{diffz}
\ee
\noindent  where it is understood that 
\be Z(\lambda) := Z(J,\eta,\eta^*,\lambda) \biggr|_{J=\eta=\eta^*=0} \, .
\ee
\noindent The last step is to rewrite Eq.(\ref{diffz}) as a differential equation  
for $M_0$  by using Eq.(\ref{mno})  to replace 
\be Z(\lambda) = Z(0) \exp M_0(\lambda) \, . \ee
\noindent Using this into Eq.(\ref{diffz}), we finally obtain Eq.(\ref{wprime}).

\end{document}